\documentclass[reqno,12pt,letterpaper]{amsart}
\usepackage{amsmath,amssymb,amsthm,graphicx,mathrsfs,url}
\usepackage[usenames,dvipsnames]{color}
\usepackage[colorlinks=true,linkcolor=Red,citecolor=Green]{hyperref}
\usepackage{amsxtra}
\usepackage{placeins}
\usepackage{wasysym} 
\usepackage{graphicx}
\usepackage{subcaption}
\usepackage{soul}
\def\arXiv#1{\href{http://arxiv.org/abs/#1}{arXiv:#1}}
\usepackage{array}
\newcolumntype{P}[1]{>{\centering\arraybackslash}m{#1}}

\let\Im=\Imag

\DeclareMathOperator{\Real}{Re}
\DeclareMathOperator{\Spec}{Spec}
\DeclareMathOperator{\tr}{tr}
\let\Re=\Real

\let\Im=\Imag 
\let\Re=\Real 

\def\?[#1]{\textbf{[#1]}\marginpar{\Large{\textbf{??}}}}

\def\smallsection#1{\smallskip\noindent\textbf{#1}.}
\let\epsilon=\varepsilon 

\newcommand{\RR}{{\mathbb R}}
\newcommand{\NN}{{\mathbb N}}
\newcommand{\CC}{{\mathbb C}}

\newcommand{\ZZ}{{\mathbb Z}}

\def\indic{\operatorname{1\hskip-2.75pt\relax l}}

\newtheorem{theo}{Theorem}
\newtheorem{prop}{Proposition}

\newtheorem{lemm}[prop]{Lemma}

\newtheorem{rem}{Remark}
\numberwithin{equation}{section}

\newcommand{\ad}{\operatorname{ad}}
\newcommand{\supp}{\operatorname{supp}}

\title{Localised Davies generators for (pseudo)differential operators}

\author{Jeffrey Galkowski}
\email{j.galkowski@ucl.ac.uk}
\address{Department of Mathematics, University College London, WC1H 0AY, UK}

\author{Maciej Zworski} 
\email{zworski@berkeley.edu}
\address{Department of Mathematics, University of California, Berkeley, CA 94720}

\begin{document}

\maketitle

\begin{abstract}
A classical Davies generator provides a Lindbladian for which the Gibbs state is stationary. Its construction involves precise knowledge of the Bohr spectrum or equivalently state evolution for all times. Recently Chen, Kastoryano and Gily\'en \cite{chen} proposed a construction involving localisation in time and carried out in the case of finite dimensional Hilbert spaces. The resulting generators are called quantum Gibbs samplers as the corresponding Lindblad evolution is expected to settle to the Gibbs state. 
In this paper, we show that the localised Davies construction also works
for natural classes of unbounded operators, including
pseudodifferential operators used in the study of classical/quantum
correspondence in Lindblad evolution.  Our emphasis is microlocal: we
prove that the localised jump operators are themselves
pseudodifferential, and hence pseudolocal.  The proof involves a novel version
of Egorov's theorem valid for all times.
\end{abstract}

\section{Introduction}
\label{s:int}

For a quantum system at a non-zero temperature $ 1/\beta $  the role of a ground state of a self-adjoint Hamiltonian, $P $, is taken by the Gibbs state $ \rho_\beta := e^{-\beta P}/{\rm{tr}} e^{-\beta P }  $. A quantum expectation value of an observable $ A $ (an operator on a Hilbert space) at temperature $ 1/\beta $ is then defined as $ \langle A \rangle_\beta := {\rm{tr}} A \rho_\beta $ and its efficient evaluation is of interest in different situations (see the references below). One possible {way to make this evaluation is to design} an interaction with an open system {so that states rapidly decay to the Gibbs state}. {That interaction} can be modelled by a Lindblad evolution generated by a super-operator  $\mathcal L $ --  see \eqref{eq:Lind} below for an example 
 and \cite{gaz5} for an introduction from a PDE perspective -- for which the Gibbs state is an equilibrium state. In that case, {the desired property is}
\begin{equation}
\label{e:settleDown} \langle A \rangle_\beta = \lim_{ t \to +\infty } {\rm{tr}}   A \, e^{ t \mathcal L} B , 
\end{equation}
for a class of operators $ B$ with $ {\rm{tr}} B = 1 $. 
A necessary condition for this to hold is 
\begin{equation}
\label{eq:Lbeta} \mathcal L \rho_\beta =0 .
\end{equation}
{Although~\eqref{eq:Lbeta} is not sufficient to imply~\eqref{e:settleDown} in general}, if we have \eqref{eq:Lbeta}, one can hope that the dissipative term in $ \mathcal L $ produces convergence to the Gibbs state. Here, {we work in the setting of differential and pseudo-differential operators as Hamiltonians and only consider the question of} finding $ \mathcal L $ so that  \eqref{eq:Lbeta} holds. 
However, for completeness, Appendix \ref{a:b} illustrates how \eqref{e:settleDown} follows for the harmonic oscillator
 -- see \eqref{eq:decayQ} for a 
quantitative version of \eqref{e:settleDown} in that case.

In various works the importance of locality of the coherent part of the Lindbladian 
(see \eqref{eq:defB0}) and of the jump operators (see $ \mathcal D_f $ in \eqref{eq:Lind})
is stressed. From the PDE point of view strict locality of linear operators acting on 
smooth functions on $ \mathbb R^n $ ($ \supp P u \subset \supp u$) holds only for differential operators. Their generalisation, pseudodifferential operators (see \eqref{eq:Weylq}) are {\em pseudolocal}. For the class
of operators \eqref{eq:Weylq} with $ a $ satisfying $ \partial^\alpha a = \mathcal O ( 1 ) $ for  $ |\alpha| \geq k $, 
we can write 
\[  A u ( x) = \int_{\mathbb R^n } K_A ( x, y ) u (y ) dy , \]
(an informal expression which should be understood in the sense of distributions -- see \cite[\S 5.2]{H1}).
Then
for every $ N \gg 1 $,  
\begin{equation}
\label{eq:umaps}  u \mapsto \int_{\mathbb R^n } |x -y|^{2N} K_A ( x, y ) u ( y ) d y \ \ 
\text{ is bounded on $ L^2 $. } 
\end{equation} 
{In this article, we construct a large family Lindbladians with pseudodifferential, and hence pseudolocal, jump operators.}
The proof of {this property} involves a novel version of Egorov's theorem on the classical/quantum correspondence -- see Theorem \ref{t:egorov} for the statement and \cite{gaz5} for an introduction to the classical/quantum correspondence in the context of Lindbladians.

We remark that for $ a $'s with better decay properties, we have a stronger statement (see 
\cite[Theorem 9.6]{z12}) 
\begin{equation}
\label{eq:pseudolo} \exists \, C_0 \, \forall \, N ,  \alpha\in \mathbb N^{2n}  \, \, \exists \, C_N \ \  | \partial_{x,y}^\alpha K_A ( x , y ) | \leq C_N  |x-y|^{-N} . \end{equation}

\subsection{Localised generators}
A now classical construction of $ \mathcal L $ satisfying \eqref{eq:Lbeta} was given by Davies \cite{dav1},\cite{dav2} and we review it in the case of matrices in \S \ref{s:Davies} following \cite[\S 3.3]{brp}. It is constructed using an essentially arbitrary family of operators, $ \mathcal A $, but it involves their evolution under $ e^{ - i t P } $ for all times -- see \S \ref{s:davlim}.

Here we are motivated by a recent paper by Chen--Kastoryano--Gily\'en \cite{chen}  where
a construction involving localisation in time was proposed and carried out in the case of 
finite dimensional Hilbert spaces. (See also Ramkumar--Soleimanifar \cite{ras} and 
Ding--Li--Lin \cite{dll} for further developments and references).
We show that under certain conditions (natural in the case of 
differential and pseudo-differential operators) the construction works in the case of unbounded operators. 
A different approach to the infinite dimensional case has also been recently proposed by Becker--Rouz\'e--Salzmann \cite{beck}.
We refer to that paper and to 
\cite{chen} for references and background in the context of quantum information. 

Compared with the finite-dimensional quantum Gibbs-sampling literature,
our emphasis is not on implementation or mixing, but on the analytic and
microlocal well-posedness of the localised Davies construction for
unbounded operators.  Compared with the recent infinite-dimensional
Dirichlet-form approach of \cite{beck}, our results are more
specialised in the direction of differential and pseudodifferential
operators: we prove stationarity of the Gibbs state for the localised
construction and show that the localised jump operators are
pseudodifferential, hence pseudolocal.  
This pseudodifferential structure is the main new feature of the present
work relative to the existing Gibbs-sampling literature.

The motivating differential example is given by,  
\begin{equation}
\label{eq:Schr}   \begin{gathered}  P = - \Delta + V ( x ) , \ \ \ V \in C^\infty ( \mathbb R^n, \mathbb R )  , \\  \partial^\alpha V ( x ) = \mathcal O ( 1 ), \ \ 
| \alpha | \geq 2 , \ \  V ( x ) \geq |x|^2 /C - C, \end{gathered} \end{equation}
and (in the construction of the Lindbladian below)
\begin{equation} 
\label{eq:diffA}   
\begin{gathered}  A = \langle a_1  , \partial_ x \rangle +  \langle a_2, x \rangle   ,  \ \ \  a_j \in \mathbb C^n .
\end{gathered}  \end{equation}

To describe the construction in \cite[Theorem I.1]{chen}, we recall the
{\em operator Fourier transform}: for a self-adjoint operator $ P $ and for $ A $ (in a class described below) we define 
\begin{equation}
\label{eq:OFT} \hat A_f ( \omega ) :=  \frac{1}{\sqrt{2 \pi} } 
\int_{\mathbb R } e^{ i P t } A e^{-i P t } e^{-i \omega t } f ( t ) dt, 
\end{equation}
where, following \cite{chen}, we will (mostly) take 
\begin{equation}
\label{eq:defsigma}  f( t ) = f_\sigma ( t ) := \sigma^{\frac12} \pi^{\frac14} e^{ - t^2 \sigma^2/2 } , \ \ \ \hat f_\sigma ( \tau ) = f_{1/\sigma } ( \tau) , \ \ 
\| f_\sigma \|_{L^2} = 1 , \end{equation}
where $ \hat g ( \tau ) := (1/\sqrt{2 \pi } ) \int g ( t ) e^{-i t \tau}dt $. 
For a self-adjoint Hamiltonian $ P $ and a set of operators (with properties specified below), $ \mathcal A $, we define
the following Lindbladian  
\begin{equation}
\label{eq:Lind}  
\begin{gathered} \mathcal L_f T  := - i [ \beta^{-1} P + B  , T ] + \mathcal D_f ( T)  \\
 \mathcal D_f ( T ) :=  \int_{\mathbb R } 
\gamma( \omega )  \sum_{A \in \mathcal A  } \left( \hat A_f ( \omega ) T \hat A_f  ( \omega )^* - \tfrac12  \left\{ \hat A_f  (\omega )^* \hat A_f  ( \omega ) , T \right\}
\right)
 d \omega \end{gathered}  \end{equation}
The additional {\em coherent term} in Lindblad evolution is defined as 
\begin{equation}
\label{eq:defB0} B:= \sum_{ A \in \mathcal A } \int_{\mathbb R } b_1 ( t ) e^{-i  P t } \left( \int_{\mathbb R } e^{ i  P s } 
A^* e^{-2 i  Ps } A e^{ i s  P }b_2(s) ds \right) e^{ i  P t } dt , \end{equation}
where $ b_1 , b_2 $ satisfy
\begin{equation}
b_1,b_2\in L^1(\mathbb{R};\mathbb{R}).
\end{equation}

As we indicate in the case of matrices (see \ref{s:davlim}), the Davies generator is the delocalisation limit, $ \sigma 
\to 0 $ (the Gaussian becomes flat) and hence we refer to the \cite{chen} generators as {\em localised 
Davies generators}. They are more commonly known as {\em quantum Gibbs samplers} but that implicitly assumes convergence to the Gibbs state in Lindblad evolution which is a question not addressed here.

\noindent
{\bf Remark.} The Davies generator appears naturally in the coupling of a system to a thermal bath:
it is the Fermi Golden Rule perturbative term -- see Derezi\'nski--Jak\v{s}i\'c \cite{deja} for the case of 
a finite dimensional system. The localised generator of Chen--Kastoryano--Gily\'en \cite{chen} can also be derived from that framework by stopping the Davies
argument before the final infinite-time secular projection \cite{Jap}.

\subsection{Statements of the results}

We now specify our assumptions on $P $ and on $ \mathcal A $. We first consider an abstract setting
and then show how additional structure gives results for pseudodifferential operators.

For $ \mathcal H $, a separable Hilbert space, 
\begin{equation}
\label{eq:assP} 
P:\mathcal{H}\to \mathcal{H} \ \text{ is an unbounded, self-adjoint operator with $P\geq 1$.}
\end{equation}
(Since adding a constant to $P$ does not change any of the objects above, for our purposes, this is equivalent to $P$ being bounded from below, a natural condition for considering $ e^{-P } $.) 
We can define (using the notation $ X' $ for the dual of $ X $, with duality defined via the inner product
on $ \mathcal H $)
\begin{equation}
\label{eq:defDs}    
\begin{gathered}  \mathscr D^s := \mathcal D ( P^s ) , \ \ \  \mathscr D^s = P^{-s} \mathcal H  , \ \ s \geq 0 , \ \ \
\mathscr D^s = ( \mathscr D^{-s} )' , \ \ s \leq 0 .
\end{gathered}  \end{equation}
We now assume that there exists $ k\in {\mathbb{R}} $ such that $ A : \mathscr D^k\to \mathcal  H $ 
and $ A^* : \mathscr D^k \to \mathcal H $. Here at first $ A^* : \mathcal H \to \mathscr D^{-k} $ 
and we require this additional mapping property. 
More precisely, we assume that 
\begin{equation}  
\label{eq:assA1} 
   \sum_{ A \in \mathcal A } \| A \|_{ \mathscr D^k \to \mathcal H  }^2  +  \| A^* \|_{ \mathscr D^k \to \mathcal H  }^2
< \infty  
\end{equation}
It is then immediate that for $ \gamma \in L^1 ( \mathbb R ; [0,\infty ) ) $, $ f, b_j \in L^1 ( \mathbb R ) $, 
$ \mathcal L_f $ defined in \eqref{eq:Lind} has the following mapping property:
\begin{equation}
\label{eq:basicm} 
\mathcal L_f : \mathcal L ( \mathscr D^{-k} , \mathscr D^k ) \to \mathcal L ( \mathscr D^{k}  , \mathscr D^{-k} ) . 
\end{equation}

When we specialise to Gaussian $ f $'s and assume in addition that $ \mathcal A $ is closed under taking the 
adjoints we  need a  {\em balance} condition on $ \gamma$ {to} guarantee that we can find $ B $
of the form \eqref{eq:defB} such that $ \exp ( - P ) $ is stationary for \eqref{eq:Lind}. The standard
{\em Kubo–Martin–Schwinger detailed balance} condition (see Benoist et al \cite{voj} for a recent 
abstract investigation of balance conditions)
 on $ \gamma $, needed in the Davies generator
\eqref{eq:gam2D}, is the limiting case as $ \sigma \to 0 $, that is, when the Gaussian in the sampling
\eqref{eq:OFT} becomes flat  (see \S \ref{s:davlim}). The balance condition \eqref{eq:balao} appeared 
already in \cite[Lemma 7.1]{ras} and  the expressions for $ b_j$'s in $ B $ are essentially the same as in \cite{chen}.

\begin{theo}
\label{t:2}
Assume \eqref{eq:assP}, \eqref{eq:assA1}, $f = f_\sigma $ in \eqref{eq:defsigma}, $ \mathcal A $ 
is closed under taking adjoints, and
\begin{equation}
\label{eq:balao}   \gamma ( \omega ) = e^{ -\omega / 2 } \varphi \left( \omega + \tfrac{1}{ 4}  \sigma^2 \right), \ \ 
\varphi ( -\omega ) = \varphi ( \omega ) \in e^{-|\omega|/2}L^1 ( \mathbb R ). 
\end{equation}
If $ B $ is defined by \eqref{eq:defB0} with 
\begin{equation}
\label{eq:defb12} 
\hat b_1(\tau)=\frac{1}{4 \sigma \sqrt \pi i}e^{- \tau^{2}/4\sigma^2}\tanh\!\left(\frac{\tau}{4}\right),
\ \ \ 
\hat b_2(\tau)=e^{-\tau/2}\int_{\mathbb R}\gamma(\omega)\,e^{-(\omega+\tau/2)^{2}/\sigma^2}\,d\omega ,
\end{equation}
(see \eqref{eq:b1b2FT} for $ b_j (t)$'s) then, in the notation of \eqref{eq:Lind}, 
$ \mathcal L_f ( e^{-P } ) = 0 $. 
\end{theo}

An important issue now is to see if $ \mathcal L_f $ generates a contraction on the trace class. 
For that we make the following assumption:
there is $ \delta > 0 $ such that
\begin{equation}
\label{e:assumeA}
 \sum_{A\in\mathcal{A}} \big(\|AP^{-\frac12}\|^2_{\mathcal{H}\to \mathcal{H}}+\|[P,A]P^{-\frac12}\|_{\mathcal{H}\to \mathcal{H}}^2+\|[P,[P,A]]P^{-1+\delta}\|_{\mathcal{H}\to \mathcal{H}}^2\big)<\infty, 
\end{equation}
where the operators are initially defined as operators $ \mathscr D^{N}  \to \mathscr D^{-2}  $, for sufficiently large $ N$. 
When the operators
in $ \mathcal A $ are bounded, it is enough to assume
\begin{equation}
\label{e:assumeAb}
 \sum_{A\in\mathcal{A}} \| A \|^2 < \infty.
 \end{equation}


The key now is that   \eqref{e:assumeA} implies that   $ \mathcal L_f $ in \eqref{eq:Lind} satisfies the Davies condition \cite{dav}:
\begin{equation}
\label{eq:Daviesc} 
\begin{gathered} Y := i ( \beta^{-1} P + B ) - \tfrac12 \sum_{A \in \mathcal A }  \int_{\mathbb R } \gamma ( \omega ) \hat A_f ( \omega )^* \hat A_f ( \omega)d\omega
\end{gathered} 
 \end{equation}
is the infinitesimal generator of a strongly continuous one parameter contraction semigroup on 
$ \mathcal H $. This may be of independent interest and it gives
\begin{theo}
\label{t:1} 
Suppose that $ \mathcal A  $ satisfies 
\eqref{e:assumeA},  $P$ satisfies \eqref{eq:assP} and 
$ \gamma \in L^1 ( \mathbb R ;[0,\infty)) $. 
Then for $ f , b_1, b_2 \in L^1$, 
$  \mathcal L_f $ is a well defined Lindbladian, in the sense that $ e^{ t \mathcal L_f }  $ is a contraction 
on the space of trace class operators $ \mathcal L_1 ( \mathcal H ) $.
\end{theo}

\noindent
{\bf Remark.} Once $ Y $ in \eqref{eq:Daviesc} generates a strongly continuous contraction semigroup, the Davies construction 
\cite{dav} implies that $\mathcal L_f$  generates a
minimal completely positive contraction semigroup on trace class.   If in addition $ e^{-P} $ is of trace class and 
stationary, then $ \exp t \mathcal L_f $ is 
trace preserving.  This is the
standard conservativity criterion -- see Chebotarev--Fagnola
\cite{Chef}. In other words, the Lindblad evolution has the desired properties standard in finite dimensions.

\subsection{Differential and pseudodifferential operators}

When we assume more structure, an analogue of Theorem \ref{t:1} is valid for pseudodifferential operators, even though \eqref{e:assumeA} may be violated. The operators we consider are in the class used in recent works on 
classical/quantum correspondence in Lindblad evolution by 
Hern\'andez--Ranard--Riedel \cite{hrr} and the authors \cite{gaz5} (see also Li \cite{kevin} and Smith \cite{hasm} for more recent progress on the PDE study of Lindblad evolution). For harmonic oscillators (which are a very special case of our operators) 
questions related to Lindbladians with stationary Gibbs states were addressed by 
Cipriani--Fagnola--Lindsay \cite{cipr} and, for localised generators, by 
\v{Smid} et al \cite{smid}, where very precise results about return to equilibrium were provided.

To state the result we recall the Weyl quantization. Let $ a \in C^\infty ( \mathbb R^{2n} ) $
(a classical observable, that is, a function on phase space) satisfy
$ |\partial^\alpha a| \leq ( 1 + |x| + |\xi| )^M $ for some $ M$ and all $ \alpha $. 
For $ u \in \mathscr S ( \mathbb R^n ) $ (the class of functions with rapidly decaying derivatives \cite[\S 3.1]{z12}) the action of the Weyl quantisation of $ a$ on $ u $ 
is given by 
\begin{equation}
\label{eq:Weylq} a^{\rm{w}}( x, D ) u = \frac{1}{ (2 \pi )^n } \int_{\mathbb R^{2n}} a \left (\tfrac12 ( x + y ) , \xi \right) u ( y ) e^{ i \langle x - y, \xi \rangle } dy d \xi . \end{equation}

We use the notation
\[ \rho = ( x, \xi ) \in \mathbb R^{2n} , \ \ \ \langle \rho \rangle = ( 1 + |x|^2 + |\xi|^2 )^{\frac12 } . \]

\begin{theo}
\label{t:3} 
Suppose that $ p \in C^\infty ( \mathbb R^{2n} ) $ satisfies 
\begin{equation}
\label{eq:asp}   | \partial^\alpha p ( \rho) | \leq C_\alpha ,  \ \ | \alpha | \geq 2, \ \  p ( \rho ) \geq \langle
\rho \rangle^2 /C_0 - C_0,  \end{equation}
$ P = p^{\rm{w}}  ( x, D ) $ and 
 $ \mathcal A = \{ a^{\rm{w}} ( x, D ) : a \in \mathcal A_{\rm{cl}} \subset C^\infty ( \mathbb R^{2n} )  \} $ where, 
 \begin{equation}
\begin{gathered} 
\label{eq:asA0} 
\sum_{ a \in  \mathcal A_{\rm{cl}} } \Big( | a ( 0 )|^2 + 
\sum_{   1 \leq |\alpha | \leq N_0 }  \sup_{\rho \in \mathbb R^{2n} }  | \partial^\alpha a ( \rho ) |^2 \Big) \leq C , \ \ \ 
\overline { \mathcal A}_{\rm{cl}}  =  \mathcal A_{\rm{cl}} ,
 \end{gathered} 
 \end{equation}
where $ N_0 $ is a (large) constant depending on the dimension. Then the assumptions of Theorem \ref{t:2} are satisfied and the conclusions of Theorem 
 \ref{t:1} hold. \end{theo}
 
 In the case of $ P$ and $ A $'s of Theorem \ref{t:3}, we have the following result about 
 pseudodifferential structure of the jump operators $ A_f ( \omega ) $ in \eqref{eq:Lind}. In view of  \eqref{eq:pseudolo} this shows that they are pseudolocal.
 
 \begin{theo}
 \label{t:4}
 Suppose that $ P = p^{\rm{w}} ( x, D ) $, $ A = a^{\rm{w}} ( x, D ) $ where
 \[  |\partial^\alpha p ( \rho )| \leq C_\alpha, \ \ |\alpha| \geq 2, \ \ \ 
 |\partial^\alpha a ( \rho ) | \leq C_\alpha, \ \ |\alpha| \geq 1 , \ \ \ \rho \in \mathbb R^n 
 \times \mathbb R^n , \]
 and $ p $ is real valued. Then, for $ f $ satisfying 
$  | f ( t ) | \leq C_N e^{ -N |t| } $, for all $ N$, 
and in the notation of \eqref{eq:OFT}, 
\[  A_f ( \omega ) = a_f ( \omega )^{\rm{w}} ( x, D ), \]
where $  | \partial^\alpha a_f (\omega,  \rho )| \leq C_\alpha$, $ |\alpha| \geq 1$, 
uniformly in $ \omega $. 
\end{theo}

This is the point at which the pseudodifferential setting gives more
than an abstract unbounded-operator extension: the localised 
jump operators inherit a microlocal locality property \eqref{eq:umaps}.

Even with the choices of $ f $ as in Theorem \ref{t:1}, we cannot exclude the possibility 
that $ B $ is not pseudodifferential if only \eqref{eq:asp} is assumed. That is due to the fact that $ b_1 (t) $ is not super-exponentially decaying in that case. It is possible that under stronger assumptions, for instance \eqref{eq:Schr}, the exponential decay of 
$ b_1 ( t ) $ is sufficient to guarantee pseudodifferential structure of $ B $.

We conclude this introduction with some examples.

\subsubsection{Examples for Theorem \ref{t:2}} 
\label{s:ex1} 
The assumptions in Theorem \ref{t:2} are very weak and apply to large classes of examples.
Here is one formulated using order functions \cite[\S 4.4,\S 8.2]{z12} used for definitions and composition rules for pseudodifferential operators. We say that $ m : \mathbb R^{2n} \to [0, \infty ) $ is an order function
if there exist $ C , N $ such that for $ X,Y \in \mathbb R^{2n} $, $ m ( X ) \leq ( 1 + |X-Y|)^N m ( Y ) $. 
Then a classical observable $ a \in C^\infty ( \mathbb R^{2n} ) $ (a function on the phase space) is
a symbol associated to $ m$, $ a \in S ( m ) $, if $ |\partial^\alpha a | \leq C_\alpha m $ for all $ \alpha \in 
\mathbb N^n $ (here $ \partial^\alpha = \partial_{x_1}^{\alpha_1} \cdots  \partial_{x_n}^{\alpha_n} $, 
$ |\alpha| = \alpha_1 + \cdots \alpha_n $. 
For $ u \in \mathscr S ( \mathbb R^n ) $ (the class of functions with rapidly decaying derivatives \cite[\S 3.1]{z12}) a quantisation of $ a \in S ( m )  $, 
is given by \eqref{eq:Weylq}.

Suppose now that $ m \geq 1 $ , 
\begin{equation}
\label{eq:asP1}  P = p^{\rm{w}}  ( x, D ) , \ \ \  p \in S ( m ) ,   \ \ \  p  \geq m/C - C , 
\end{equation}
 We assume that $ P $ is self-adjoint with the domain given by the generalised Sobolev space $ H(m) $ (see \cite[8.2]{z12}; we consider the case of $ h = 1 $). We also assume that $ P \geq 1 $ which implies that $ P^{-1} \in S ( m^{-1} ) $ (see the Appendix). We then have 
$ \mathscr D^s = H (m^s )  $. The condition \eqref{eq:assA1} holds provided that
\begin{equation}
\label{eq:assA2} 
\begin{gathered}
\mathcal A = \{ a^{\rm{w}} ( x, D ) : a \in \mathcal A_{\rm{cl}} \} , \ \ \ 
\sum_{ a \in  \mathcal A_{\rm{cl}} } \sum_{ |\alpha | \leq N_0 } \sup_{\mathbb R^{2n} } | m^{-k} \partial^\alpha a |^2 < \infty , \ \ \  \overline { \mathcal A}_{\rm{cl}}  =  \mathcal A_{\rm{cl}} , 
\end{gathered}
\end{equation}
where $ N_0$ is some (large) fixed constant depending only on $ n$. 

We now list some concrete examples which fit into this framework: the condition $ P \geq 1 $ can 
be obtained by adding a constant to $ P$. One was already given in Theorem \ref{t:3}. 

\noindent {\bf Example 1.} Consider $ \mathcal A_0 \subset
C^\infty ( \mathbb R^n ; \mathbb C^{n+1} )  $, and
\[ \begin{gathered}  
  P = - \Delta + V ( x ) , \ \  | \partial^\alpha V ( x ) | \leq C_\alpha,  \\
\mathcal A  = \{ \langle a_1 ( x ), \partial_x \rangle + a_2 ( x ) : (a_1, a_2 )  \in \mathcal A_0 \}, \ \ \ 
  \sum_{ a \in \mathcal A_0 }  \sum_{ |\alpha | \leq N_0 } \sup_{\mathbb R^n}  | \partial^\alpha a |^2 < \infty  , 
\end{gathered} \]
and that, (to have the set closed under complex conjugation), 
\[ (a_1, a_2 )  \in \mathcal A_0 \ \Longrightarrow \   ( -\bar a_1 , - {\rm{div}} \, \bar a_1 + \bar a_2 ) \in \mathcal A_0 . \]

In this case, $ m = ( 1 + | \xi |)^2 $ and $ H (m^s ) = H^{2s} ( \mathbb R^n )$ (the usual Sobolev space)
and \eqref{eq:assA2} holds with $ k = \frac12 $. Note that in this case $ e^{-P} $ is {\em not} of trace class.

\noindent {\bf Example 2.} Suppose $ (M,g )  $ is a compact Riemannian manifold and
$ P = - \Delta_g + V ( x )$, $  V \in C^\infty ( M, \mathbb R  )$,  and $ \mathcal A $ a family of vector fields with coefficients
bounded in $ C^{N_0} ( M ) $. This could be generalised of $ P $ being any self-adjoint elliptic operator with a non-negative principal symbol of order $ m $ and $ \mathcal A $ a family of pseudodifferential operators of order $ k/2 $
with boundedness condition similar to \eqref{eq:assA2} -- see \cite[\S 14.2]{z12} for operators on manifolds.

\subsubsection{Examples for Theorem \ref{t:1}}
\label{s:ex2}

The assumption \eqref{e:assumeA} is not valid for operators in Theorem \ref{t:3}. 
We can however consider some cases in which it is valid and to which Theorem \ref{t:3} does not apply.

The first example generalises \eqref{eq:Schr} slightly:

\noindent
{\bf Example 3.}  Consider $ P = - \Delta + V ( x )  $ satisfying 
\[   |\partial^\alpha V ( x ) | \leq \left\{ \begin{array}{ll} C ( 1 + |x| ) &  |\alpha| =1, \\
 C_\alpha ( 1 + |x| )^{1-2\delta} , &  |\alpha | \geq 2 , \end{array} \right. \ \ 
V ( x) \geq   |x|^2/C - C , 
\]
(for instance $ V ( x ) = x^2 + ( 1 + x^2 )^{\frac12 ( 1 - 2 \delta )} \cos x $, $ x \in \mathbb R $)
with 
\[  \mathcal A = \{  \langle a_1, \partial_x \rangle + \langle a_2 , x \rangle : ( a_1, a_2 ) \in \mathcal A_0 
\subset \mathbb C^{2n} \}, \ \ \ \sum_{ a \in \mathcal A_0 } | a |^2 < \infty . \]

\noindent
{\bf Example 4.} Consider the torus $   \mathbb T^n := \mathbb R^n/\mathbb Z^n $ and an elliptic self-adjoint 
(on $ L^2 ( \mathbb R^n , dx ) $) second order operator 
\[ \begin{gathered}   P = - \sum_{ 1\leq i, j \leq n } \partial_{x_i} p_{ij} ( x ) \partial_{x_j } + V ( x ) , \ \ \ \ 
p_{ij} \in C^\infty (\mathbb T^n ) ,   \ \  V \in C^\infty ( \mathbb T^n ) , \\
 p_{ij} (x) = \overline{ p_{ji} (x) }, \ \ \ \sum_{ i, j} p_{ij} ( x ) \xi_i \xi_j \geq |\xi|^2 / C , \ \ \xi \in \mathbb R^n .
\end{gathered}\]
Then the following family works:
\[ \begin{gathered} 
\mathcal A = \left\{  \langle a , \partial_x \rangle + b ( x )  : a \in \mathcal A_0 \subset  \mathbb C^n , \ 
b \in \mathcal B \subset C^\infty ( \mathbb T^n ) \right\} 
, \\ \sum_{ 
a \in \mathcal A_0 } |a|^2 + \sum_{ \mathcal B } \sum_{ |\alpha | \leq N_0 } 
\sup |\partial^\alpha b |  < \infty . \end{gathered} \]
We should note that the operators $ u \mapsto b u $ do not satisfy \eqref{e:assumeA} but they are bounded operators and hence can be handled directly.

\smallsection{Acknowledgements} 
JG acknowledges support from EPSRC grants EP/V001760/1 and EP/V051636/1, the Leverhulme Trust under Research Project Grant  RPG-2023-325, and the ERC under the Synergy grant PSINumScat 101167139, and 
MZ from
the Simons Foundation under a ``Moir\'e Materials Magic" grant. We are also grateful to Anthony Chen for 
introducing us to this subject, comments on the first version of the paper and, in particular, reference \cite{ras}, Simon Becker for references \cite{beck},\cite{cipr},\cite{smid} and many useful insights,  Simon Horsley for a helpful discussion on the role of Lindbladians in physics, and Zhen Huang,  Lin Lin, and Vojkan Jak\v{s}i\'c for additional references. We are particularly grateful to Zhen Huang for pointing out a sign mistake in our balance condition! 
We also acknowledge the help of ChatGPT-5 with the computations in \S \ref{s:func}, Appendix \ref{a:b}, and 
the proofreading. 
The second author would also like to thank University College London, for providing support and hospitality
on a visit during which this note was written.


\section{The matrix case revisited}
\label{s:matr}

We start by presenting a different perspective on the case of
\begin{equation}
\label{eq:Hfin}   P = P^* \in \mathcal L ( \mathcal H ) , \ \ \  \mathcal A \subset \mathcal L ( \mathcal H ), \ \ \ \dim \mathcal H < \infty . 
\end{equation}
The functional equation which determines the properties of $ \gamma $ and the functions $ b_j $ in \eqref{eq:defB0} will be the same 
in the infinite dimensional case but it can be presented in a straightforward way. 
We start with the description of the Davies generator \cite{dav1},\cite{dav2} as presented in \cite{brp}.
We then show how the Davies generator arises as the limit of 
the localised generator. This is followed by analysis of the localised version 
from \cite{chen}. 

\subsection{The Davies generator}
\label{s:Davies}
Following \cite[\S 3.3]{brp} (see also \cite[equation (4)]{tem}), we introduce the Davies generator. This generator provides a Lindbladian for which
the Gibbs state, $ \rho = e^{-P}/{\rm{tr}} e^{-P } $, is stationary. It is constructed as follows. Define
\begin{equation}
\label{eq:Anu}    A_\nu := \sum_{ E_2 - E_1 = \nu } \indic_{E_2} ( P ) A \indic_{E_1 } ( P ) , \ \ \ 
[ P, A_\nu] =  \nu A_\nu,  \ \ \  e^{ P } A_\nu e^{-P} = e^{\nu} A_\nu .
\end{equation}
The index $ \nu $ varies along the Bohr spectrum of $ P $:
\[ \mathcal B ( P ) :=  \{ E_2 - E_1 : E_j \in \Spec( P ) \}. \]
A version of the Davies generator is now defined for a family $ \mathcal A \subset \mathcal L ( \mathcal H )$
as
\begin{equation}
\label{eq:Davies}   \mathcal D (T ) := \sum_{ A \in \mathcal A } \sum_{\nu \in \mathcal B ( P ) } 
\gamma ( \nu ) \left( A_\nu T A_\nu^* - \tfrac12 ( A_\nu^* A_\nu T + T A_\nu^* A_\nu ) \right) . \end{equation}
We now have
\begin{prop}
\label{p:Davies}
Suppose that $ \mathcal D $ is given by \eqref{eq:Davies}, $ \mathcal A $ is closed under taking adjoints, 
and 
\begin{equation}
\label{eq:sumanorm}
\sum_{ A \in \mathcal A } \| A \|^2 < \infty . 
\end{equation}
If $ \gamma $ satisfies the following {\em balance condition}
\begin{equation}
\label{eq:gam2D}  \gamma  ( - \omega ) = e^{ \omega}   \gamma ( \omega ) \ \Longleftrightarrow \ 
 \gamma ( \omega ) = e^{-\frac12 \omega } \varphi ( \omega ) , \ \ \varphi ( \omega ) = \varphi ( - \omega ) , 
\end{equation}
then 
\[  \mathcal D ( e^{ - P } ) = 0 . \]
In particular $ e^{-P}$ is stationary for the Lindbladian $ \mathcal L T :  = - i [ \beta^{-1} P, T] + \mathcal D ( T ) $. 
\end{prop}
\begin{proof} From \eqref{eq:Anu} we obtain
\[   A_\nu e^{-P} A_\nu^* = e^{ \nu }  A_\nu A_\nu^*  e^{-P} , \ \ 
e^{-P} A_\nu^* A_\nu = A_\nu^* A_\nu e^{-P } , \]
which means that we need to check that
\[  \sum_{ A \in \mathcal A } \sum_{\nu \in \mathcal B ( P ) } 
\left( \gamma ( \nu ) e^{\nu } A_\nu A_\nu^* -   \gamma ( \nu ) A_\nu^* A_\nu \right) = 0. \]
We now observe that $ A_{-\nu} = (A^*)_\nu^* $  so that the balance condition \eqref{eq:gam2D}
shows that the sum
is equal to 
\[ \sum_{ A \in \mathcal A } \sum_{\nu \in \mathcal B  ( P ) } \gamma( \nu )\left(  A_{-\nu} A_{-\nu} ^* 
- A_\nu^*  A_\nu \right) =  \sum_{ A \in \mathcal A } \sum_{\nu \in \mathcal B ( P )} \gamma( \nu )\left(  (A^*)_{\nu}^* 
(A^*)_{\nu}  - A_\nu^*  A_\nu \right) .
\]
Since $ \mathcal A $ is assumed to be invariant under taking adjoints the sum vanishes.
\end{proof}

\subsection{Davies generator as a limit of localised generators}
\label{s:davlim}

The following proposition relates the Davies generator from \S \ref{s:Davies} to the Lindbladian described in \S \ref{s:int}  under the assumption that $ \dim \mathcal H < \infty $ is fixed. Since we present this only to 
relate the two generators we do not consider uniformity in the dimension.  We make a somewhat strong assumption on $ \gamma $ -- it is satisfied by Gaussians considered in \cite{chen}. 
\begin{prop}
\label{p:loc2Da}
Suppose that \eqref{eq:sumanorm} holds and that
\begin{equation}
\label{eq:assgam}   \gamma ( \omega )\in C^\infty ( \mathbb R )  , \ \ \  \widehat { e^{ - \omega } \gamma ( \omega ) }  \in L^1 ( \mathbb R ) .
\end{equation}
 For  $\mathcal L_f $ and $ \mathcal D $ defined in \eqref{eq:Lind} (with $ f_\sigma $ and $ b_j$'s in \eqref{eq:defsigma} and 
\eqref{eq:defb12}) and \eqref{eq:Davies} respectively, 
\begin{equation}
\label{eq:Lpconv}  \lim_{\sigma \to 0 }\|  \mathcal L_{f_\sigma } ( T ) -  \mathcal D ( T ) \|_{\mathcal L_p ( \mathcal H ) }  =0  , \ \ \ T \in \mathcal L_p ( \mathcal H ) ,  \ \ \ 1 \leq p \leq \infty ,  \end{equation}
where $ \mathcal L_p ( \mathcal H ) $ denotes the $p$--Schatten class (see \cite[Chapter 10, \S 1.3]{kato}).
\end{prop}
\begin{proof}
In the case of matrices (since $ \mathcal B ( P) $ is finite),  definitions \eqref{eq:OFT} and \eqref{eq:Anu}  give
\[ \widehat A_{f_\sigma} ( \omega) = \sum_{\nu } A_\nu \hat f_\sigma ( \omega - \nu ) ,\]
and hence
\[ \mathcal D_{f_ \sigma }  (T ) - \mathcal D ( T)   = 
 \sum_{ A \in \mathcal A } 
  \sum_{\nu,\nu'}   \left( A_\nu T A_{\nu'}^* - \tfrac12 \left\{ A_\nu^* A_{\nu'}  , T \right\} \right) 
  ( G_\sigma ( \nu, \nu' ) - \gamma ( \nu ) \delta_{\nu\nu' } ) .
\]
The formula \eqref{eq:defsigma} gives
 \[ \begin{split} G_\sigma ( \nu, \nu' ) & := \int_{\mathbb R } \gamma ( \omega ) f_{1/\sigma} ( \omega - \nu )  f_{1/\sigma } ( \omega - \nu' ) d \omega \\
 & = e^{ - ( \nu - \nu' )^2 / 4 \sigma^2 }  \frac{\sqrt \pi } \sigma \int_{\mathbb R } \gamma (\omega ) e^{ - \sigma^{-2} ( \omega - \frac12( \nu+\nu' ) )^2 } d \omega \to   \gamma ( \nu ) \delta_{\nu \nu'} .
 \end{split} \]
 Since $ \| A_\nu T A_{\nu'}^* \|_{\mathcal L_p} \leq \|A \|^2 \|T \|_{ \mathcal L_p } $, with similar estimates for other terms, the finiteness of the sum over $ \nu $ and $\nu' $ 
 shows that for $ T \in \mathcal L_p $, $ \mathcal D_{f_\sigma } ( T ) \to \mathcal D ( T ) $ in $ \mathcal L_p $ 
 as $ \sigma \to 0 $. 
 
 It remains to show that
 $ i [ B , T ] \to 0 $ as $ \sigma \to 0 $ and we do it by showing that $ \| B \| \to 0 $.  For that we go to the inverse Fourier transforms of the $ b_j $'s appearing in  \eqref{eq:defB0} and given in \eqref{eq:b1b2FT}. 
 We write $ b_1 $ as follows 
 \[  b_1 ( t )   = \tfrac18 \pi^{\frac12}  \int_0^\infty \frac{ e^{ - \sigma^2 ( t + s )^2 } - e^{ - \sigma^2 ( t - s )^2 } }
 {\sinh ( 2 \pi s ) } ds . \]
 Since $  ( t + s )^2 \geq ( t - s )^2 $ for $ t \geq 0 $ with the opposite inequality for $ t \leq 0 $, 
 \[  \begin{split} \| b_1 \|_{L^1} & =  \tfrac14 \pi^{\frac12} \int_0^\infty \! \! \int_0^\infty 
  \frac{ e^{ - \sigma^2 ( t - s )^2 } - e^{ - \sigma^2 ( t + s )^2 }}
 {\sinh ( 2 \pi s ) } ds dt = \tfrac14 \pi^{\frac12} \int_0^\infty \! \! \int_{-s}^s \frac{ e^{-\sigma^2 u^2 }  }
 {\sinh ( 2 \pi s ) } du ds \\
 & \leq \tfrac12 \pi^{\frac12} \int_0^\infty \frac{ s } 
 {\sinh ( 2 \pi s ) } ds = \tfrac{1}{32} \pi^{\frac12} . \end{split}  \]
  For $ b_2 $ we use the assumption \eqref{eq:assgam} 
noting that 
 $ \hat \gamma ( t - i ) $ in \eqref{eq:b1b2FT} is the Fourier transform of $ e^{ - \omega } \gamma ( \omega) $.
Hence,
 $ \| b_2 \|_1 \leq  \sqrt \pi \sigma e^{ \frac14 \sigma^2 } \| \hat \gamma (  \bullet - i ) \|_1 = 
  \mathcal O ( \sigma) ,  \ \ \sigma \to 0$. 
This and \eqref{eq:sumanorm} show that $ \| B \| \to 0 $ as $ \sigma \to 0 $, completing the proof.
\end{proof}

\subsection{A localised version}
\label{s:localmat}

We now consider \eqref{eq:OFT} in the case of matrices. For that we define
\[   \begin{split} \hat A (\omega) & := \frac{1}{\sqrt {2 \pi } } \int_{\mathbb R } e^{ i t P } A e^{- i t P } e^{- it \omega } dt \\
& =
\sqrt { 2 \pi } \sum_{\nu \in B ( P ) } A_\nu \delta ( \nu - \omega ) ,   \ \ \ \ 
\hat A \in \mathscr S' ( \mathbb R ; \mathcal L ( \mathcal H ) ), \end{split} 
 \]
where $ \mathcal H $ is the finite dimensional Hilbert space on which the operators act. 

The operator in \eqref{eq:OFT} becomes (specifying the dependence on $ f $ in the notation)
\[ \hat A_f  ( \omega ) =  \frac1 { \sqrt { 2 \pi } } \hat A * \hat f ( \omega ),  \ \ 
\hat f ( \omega ) :=  \frac{1}{\sqrt {2 \pi } } \int_{\mathbb R } f ( t ) e^{ -i t \omega } dt . \]
In view of \eqref{eq:Anu}, 
\begin{equation}
\label{eq:eH2eH} e^{-P }  \hat A (\omega ) = e^{-\omega } \hat A (\omega ) e^{-P } , \ \ \ 
\hat A (\omega)^* = \widehat { A^* } (-\omega) . 
\end{equation} 
We now write 
$  \mathcal D_f ( e^{-P} )  = \frac{1}{ 2 \pi } \sum_{ j } I_j $, 
where, assume that $ \hat f $ is real valued, and that $ \mathcal A $ is invariant under taking adjoints,
\[ \begin{split} I_1 & :=  \sum_{A \in \mathcal A } 
\int_{\mathbb R^3}   \gamma( \omega )  \hat A(\tau) e^{-P } \widehat {A^*}(-\tau')  \hat f ( \omega - \tau )  \hat f ( \omega - \tau' )
d \tau d \tau' d \omega  \\
& = \int_{\mathbb R^2} 
\sum_{A \in \mathcal A } 
 \hat A(\tau) \widehat {A^*}(-\tau') e^{-P } \left( 
 \int_{\mathbb R } \gamma ( \omega ) e^{ \tau' }  \hat f ( \omega - \tau )  \hat f ( \omega - \tau' ) d \omega 
 \right) d \tau d \tau' \\
 & =  \int_{\mathbb R^2} 
\sum_{A \in \mathcal A } 
 \hat A (\tau)^* \hat A (\tau' ) e^{-P } \left( 
 \int_{\mathbb R } \gamma ( \omega ) e^{ -\tau' }  \hat f ( \omega + \tau )  \hat f ( \omega + \tau' ) d \omega 
 \right) d \tau d \tau',
\end{split} \]
\[ \begin{split} I_2 & := - \sum_{A \in \mathcal A }  \tfrac 12 \int_{\mathbb R^3}   \gamma( \omega )
\widehat {A^*}(-\tau) \hat A(\tau' ) e^{ -P }   \hat f ( \omega - \tau ) \hat f ( \omega - \tau' )  d \omega  d \tau d \tau' 
  \\
& = \sum_{A \in \mathcal A } \int_{\mathbb R^2 } \hat {A}(\tau)^* \hat A( \tau' ) e^{ -P } 
\left( - \tfrac12  \int_{\mathbb R } \gamma ( \omega )  \hat f ( \omega - \tau ) \hat f ( \omega - \tau' )  d \omega  \right) d \tau d \tau'  ,
\end{split} \]
and
\[ \begin{split} I_3 & := - \sum_{A \in \mathcal A }  \tfrac 12 \int_{\mathbb R^3}   \gamma( \omega )
 e^{-P} \widehat {A^*} (\tau) \hat A (- \tau' )  \hat f ( \omega + \tau ) \hat f ( \omega + \tau' ) 
 d \omega d \tau d \tau' \\
 & = \sum_{A \in \mathcal A } \int_{\mathbb R^2 } \widehat {A^*}(\tau) \hat A (- \tau' ) e^{ -P } 
 \left( - \tfrac12  \int_{\mathbb R } \gamma( \omega ) e^{\tau' - \tau }  \hat f ( \omega + \tau ) \hat f ( \omega + \tau' )  d \omega  \right) d \tau d \tau' \\
& = \sum_{A \in \mathcal A } \int_{\mathbb R^2 } \hat {A}(\tau)^* \hat A ( \tau' ) e^{ -P } 
 \left( - \tfrac12  \int_{\mathbb R } \gamma( \omega ) e^{\tau - \tau' }  \hat f ( \omega - \tau ) \hat f ( \omega - \tau' )  d \omega  \right) d \tau d \tau' .
 \end{split} \]
If we define
\begin{equation}
\label{eq:defF}  F ( \tau, \tau' ) :=  \int_{\mathbb R } \gamma ( \omega ) \left( 
e^{-\tau' } \hat f ( \omega + \tau ) \hat f ( \omega + \tau' ) - \tfrac12 ( 1 + e^{ \tau - \tau' } ) \hat f ( \omega - \tau ) 
\hat f ( \omega - \tau' ) 
\right) d \omega,  \end{equation}
the formulas for $I_j $'s show that 
\begin{equation}
\label{eq:Dfe} \mathcal D_f ( e^{-P} ) =  \frac{1}{ 2 \pi }  \sum_{ A \in \mathcal A }  \int_{\mathbb R^2 }  F ( \tau, \tau' ) \hat {A}(\tau)^*
 \hat A ( \tau' ) e^{ -P } 
d \tau d \tau' . \end{equation}
We now want to find a self-adjoint operator on $ \mathcal H$,  $ B $, such that, for $ \omega $ and $ b $ with suitable properties, 
\begin{equation}
\label{eq:defB}  \mathcal D_f (e^{-P} ) = i [ B, e^{-P} ] , \ \ \ 
B = \frac{1}{ 2 \pi }   \sum_{ A \in \mathcal A }  \int_{\mathbb R^2 } b ( \tau, \tau' ) \hat {A}(\tau)^* \hat A (\tau' ) 
d \tau d \tau' .
\end{equation}
Since
\[ \begin{split}  B^* & = \frac{1}{ 2\pi}  \sum_{ A \in \mathcal A }  \int_{\mathbb R^2 } \overline{ b ( \tau, \tau' ) } \hat {A}(\tau')^* \hat{ A}  (\tau ) 
d \tau d \tau' ,
\end{split}
\]
we have $ B = B^* $ (for any choice of $ \mathcal A $) if  
\begin{equation}
\label{eq:condb}  \overline{ b ( \tau , \tau') } = b ( \tau', \tau ) .
\end{equation}

We now compute the commutator using \eqref{eq:eH2eH} and 
\[ e^{ -P} \hat A ( \tau )^* = e^{-P} \widehat { A^* } ( - \tau ) = e^{ \tau } \widehat { A^* } ( - \tau ) e^{-P} =
e^{\tau} \hat A( \tau )^* e^{-P } , \]
to obtain
\[ \begin{split} [ B , e^{-P} ] & = \frac{1}{ 2 \pi }   \sum_{ A \in \mathcal A }  \int_{\mathbb R^2 } b ( \tau, \tau' )
(  \hat {A}(\tau)^* \hat A (\tau' ) e^{-P} - e^{-P} \hat {A}(\tau)^* \hat A (\tau' ) ) 
d \tau d \tau' \\
& = \frac{1}{ 2 \pi }   \sum_{ A \in \mathcal A }  \int_{\mathbb R^2 } b ( \tau, \tau' ) ( 1 - e^{ \tau- \tau' } ) 
 \hat {A}(\tau)^* \hat A (\tau' ) e^{-P} d \tau d \tau'  . \end{split} \]
  In the notation of \eqref{eq:Dfe} and \eqref{eq:defB} this leads to the following functional equation:
\begin{equation}
\label{eq:funct}  F ( \tau, \tau' ) = i( 1 - e^{ \tau - \tau' } ) b ( \tau, \tau' ) . \end{equation}
It turns out that for $ f $ given by the Gaussian centred at $ 0 $, a condition on the function $ \gamma $ is sufficient
for obtaining a solution -- see \S \ref{s:func}. 

Before solving the functional equation, we relate the form of $ B $ in \eqref{eq:defB} with that from \cite{chen} presented in 
\eqref{eq:defB0}.
\begin{lemm}
\label{l:B2B}
The operator $ B $ in \eqref{eq:defB} is given by the formula \eqref{eq:defB0} if
\begin{equation}
\label{eq:b2b12}  b ( \tau, \tau' ) = 2 \pi  \hat b_1  ( \tau' - \tau ) \hat b_2 ( \tau + \tau' ) . 
\end{equation}
\end{lemm}
\begin{proof}
Make the following change of variables in \eqref{eq:defB},  $ w = s - t $, $ v = - s - t $, so that the integral becomes
\[    \tfrac12  \int_{\mathbb R^2} b_1 ( - \tfrac12 ( v + w ) ) b_2 ( \tfrac12 ( w - v ) ) e^{i w P } A^* e^{ - i w P } e^{ i v P} 
A e^{ - i v P } d v d w . \]
To apply Plancherel's theorem we calculate
\[ \begin{split}  \tfrac12 \int b_1 ( - \tfrac12 ( v + w ) ) b_2 ( \tfrac12 ( w - v ) )e^{ - i w \tau - i v \tau' } dv dw & = 
\int b_1 ( t ) b_2 ( s ) e^{ -i  ( s- t ) \tau + i ( s + t) \tau' } dt ds \\
& =
\hat b_1 ( - \tau - \tau' ) \hat{b}_2 ( \tau - \tau' ) . \end{split} \]
Hence Plancherel's theorem shows that $ B $ given in \eqref{eq:defB0} can be written as 
\[ \begin{split} B & = \sum_{ A \in \mathcal A } \int_{ \mathbb R^2} \hat b_1 ( - \tau - \tau' ) b_2 ( \tau - \tau' ) 
\widehat { A^*} ( - \tau ) \hat A ( - \tau' ) d \tau d \tau' 
\\
&= \sum_{ A \in \mathcal A } \int_{ \mathbb R^2} \hat b_1 (  \tau'- \tau ) b_2 ( \tau + \tau' )  \widehat{A} ( \tau )^* A ( \tau' ) 
d \tau d \tau' ,
\end{split} \]
which gives \eqref{eq:b2b12}.
\end{proof}

\subsection{Solving functional equation \eqref{eq:funct} in the Gaussian case}
\label{s:func}
A necessary and sufficient condition for the existence of a function
$b(\tau,\tau')$ such that \eqref{eq:funct} holds is
$
F(\tau,\tau)=0$ for all $\tau\in\mathbb R$. 
Equivalently,
\begin{equation}
\label{eq:defF2}
\begin{split}
0
&=
\int_{\mathbb R}\gamma(\omega)
\Bigl(
e^{-\tau}\hat f(\omega+\tau)^2-\hat f(\omega-\tau)^2
\Bigr)\,d\omega
\\
&=
\int_{\mathbb R}\hat f(\eta)^2
\Bigl(
e^{-\tau}\gamma(\eta-\tau)-\gamma(\eta+\tau)
\Bigr)\,d\eta
\\
&=
\sigma^{-1}\sqrt{\pi}\int_{\mathbb R}
\Bigl(
e^{-\tau}\gamma(\eta-\tau)-\gamma(\eta+\tau)
\Bigr)e^{-\eta^2/\sigma^2}\,d\eta ,
\end{split}
\end{equation}
where in the last equality we used $f=f_\sigma$ from \eqref{eq:defsigma}, so that
$ \hat f_\sigma(\eta)^2=\sigma^{-1}\sqrt{\pi}\,e^{-\eta^2/\sigma^2}$. 
If we define 
$ G_\sigma(x):=\sigma^{-1}\sqrt{\pi}\,e^{-x^2/\sigma^2}$, 
then 
condition \eqref{eq:defF2} may be written as
\[
(\gamma*G_\sigma)(\tau)=e^{-\tau}(\gamma*G_\sigma)(-\tau)= (e^{\bullet}[\gamma*G_\sigma](\bullet))(-\tau).
\]

Passing to the Fourier transform side, and assuming the required analyticity of
$\hat\gamma$, we obtain
\[
\hat\gamma(t)\widehat G_\sigma(t)
=
\hat\gamma(-t + i)\widehat G_\sigma(-t + i).
\]
Since
\[
\widehat G_\sigma(s)=2^{-1/2}e^{-\sigma^2 s^2/4},
\]
this is equivalent to
\[
\hat\gamma(t)
=
e^{i\sigma^2 t/2}e^{\sigma^2/4}\,\hat\gamma(-t +i).
\]
Equivalently, if we define
\[
\hat\varphi(t):=e^{-i\sigma^2 t/4}\hat\gamma\!\left(t + \tfrac{1}{2}i \right)
\]
then
\[
\hat\varphi(t)=\hat\varphi(-t).
\]
Taking inverse Fourier transforms, we obtain the following description of all
$\gamma$ for which \eqref{eq:defF2} holds:
\[
\gamma(\omega)=e^{-\omega/2}\,\varphi\!\left(\omega + \tfrac14 {\sigma^2} \right),
\qquad \varphi(\omega)=\varphi(-\omega).
\]

To obtain the expressions for $b_1$ and $b_2$ in \eqref{eq:b2b12}, it is convenient
to introduce
$
\xi:=\tau-\tau'$, $ \zeta:=\tau+\tau'$, 
and
\[
A(\zeta):=\int_{\mathbb R}\gamma(\omega)e^{-\omega^2/\sigma^2}e^{-\zeta\omega/\sigma^2}\,d\omega .
\]
A direct computation then gives
\[
F(\tau,\tau')
=
\sigma^{-1}\sqrt{\pi}\,
e^{-(\zeta^2+\xi^2)/4\sigma^2}
\left(
e^{-(\zeta-\xi)/2}A(\zeta)-\tfrac12(1+e^\xi)A(-\zeta)
\right).
\]
Under the divisibility condition $F(\tau,\tau)=0$, equivalently
\[
A(-\zeta)=e^{-\zeta/2}A(\zeta),
\]
this becomes
\[
F(\tau,\tau')
=
\sigma^{-1}\sqrt{\pi}\,i(1-e^\xi)\,\beta_1(\xi)\,\beta_2(\zeta),
\]
where
\[
\beta_1(\xi)=\frac{1}{2i}e^{-\xi^2/4\sigma^2}\tanh\!\left(\frac{\xi}4\right),
\qquad
\beta_2(\zeta)
=
e^{-\zeta/2}\int_{\mathbb R}\gamma(\omega)\,e^{-(\omega+\zeta/2)^2/\sigma^2}\,d\omega .
\]

In particular, we may write
\[
b(\tau,\tau')=2\pi\,\hat b_1(\xi)\hat b_2(\zeta),
\]
with
\[
\hat b_1(\xi)=\frac{1}{4\sigma\sqrt{\pi}\,i}\,
e^{-\xi^2/4\sigma^2}\tanh\!\left(\frac{\xi}{4}\right),
\qquad
\hat b_2(\zeta)=
e^{-\zeta/2}\int_{\mathbb R}\gamma(\omega)\,e^{-(\omega+\zeta/2)^2/\sigma^2}\,d\omega .
\]
%
%
%
%
This gives Theorem \ref{t:2} in the case of matrices. Here are the expressions $ b_j $'s after taking the inverse Fourier transform (less clean in the case of $ b_1 $):
\begin{equation}
\label{eq:b1b2FT}  
\begin{gathered} 
b_1(t)
=
-\tfrac{1}{8} \pi^{-\frac12}  ( e^{-\sigma^2 \bullet ^2} * \sinh(2\pi \bullet)^{-1} )(t), \ \ \ 
b_2(t) = 2\sqrt{\pi}\,\sigma\, e^{-\frac14 \sigma^2\left(2t + i \right)^2} \,\hat{\gamma}(2t + i),
\end{gathered}
\end{equation}
where $ \sinh ( 2 \pi x )^{-1} $ is considered the distribution defined by taking the principal value at 
$ 0 $: $  \sinh ( 2 \pi \bullet )^{-1} ( \varphi ) := \lim_{\varepsilon \to 0 } \int_{\mathbb R \setminus
( -\varepsilon , \varepsilon ) } \sinh ( 2 \pi t ) ^{-1} \varphi ( t ) dt $, $ \varphi \in \mathscr S ( \mathbb R ) $.

\section{Unbounded operators}
\label{s:gen}

The assumptions for Theorem \ref{t:2} are very general as illustrated in \S \ref{s:ex1}. The proof is also
straightforward since we assume that the operators $ A $ act on spaces defined using $ P $ -- see 
\eqref{eq:defDs}  and \eqref{eq:assA1}. We also note that for those spaces 
$ e^{-P } : \mathscr D^{-N } \to \mathscr D^{N } $ for any $ N $.

\begin{proof}[Proof of Theorem \ref{t:2}] The strategy is to mimic the proof in \S \ref{s:localmat} 
considering $ A ( t) := e^{it P } A e^{- i t P } $ as an operator-valued tempered distribution. 
Then $ \hat A ( \omega ) $ can be considered as the Fourier transform of $ A (t )$ in the 
distributional sense (see \cite[\S 3.2]{z12} for a brief introduction), and for $ f \in \mathscr S ( \mathbb R) $
$ A_f ( \omega ) = \hat A * f ( \omega ) $.

More precisely, using \eqref{eq:assA1}
$$ A ( t )  \in L^\infty ( \mathbb R_t ; \mathcal L (  \mathcal H, \mathscr D^{-k } )) \subset 
\mathscr S ' (  \mathbb R ; \mathcal L (  \mathcal H , \mathscr D^{-k} ) ), $$
and, in the sense of the distributions 
$\hat A (g ) := A ( \hat g )$, $ \hat A (\omega ) \in \mathscr S ' (  \mathbb R_\omega
 ; \mathcal L (  \mathcal H, \mathscr D^{-k } ) ) $. Here $ A( \hat g ) $ and $ \hat  A( g ) $
 denotes distributional pairing, informally written as
 \begin{equation}
 \label{eq:defft} A (\hat g ) = \int_{\mathbb R } A ( t ) \hat g ( t ) dt = \int_{\mathbb R } \hat A ( \omega ) g ( \omega )  d\omega = \hat A ( g ) , \ \ \ 
\hat g ( t) := \frac{1}{ \sqrt{2 \pi} } \int g (\omega ) e^{-i t \omega } d \omega . \end{equation}
We also have the corresponding statements for $ A^* $. 

The key fact now is the distributional analogue of \eqref{eq:eH2eH}:

\begin{lemm}
\label{l:exp}
Suppose that $ g \in \mathscr S ( \mathbb R $), $ \hat g  $ is holomorphic in $ \mathbb R + i (-1-\varepsilon, \varepsilon ) $, $ \varepsilon > 0 $ and 
\begin{equation}
\label{eq:expf}  | \hat g  ( t - i \mu ) | \leq C \langle t \rangle^{-2}, \ \   \mu \in [0, 1 ] . \end{equation}
Then 
\begin{equation}
\label{eq:ePAg}  e^{ - P } \hat A ( g ) = \hat A ( e^{-\bullet} g (\bullet ) ) e^{-P} \in \mathcal L (\mathcal H) 
\end{equation}
\end{lemm}
\begin{proof}
We proceed using the definition \eqref{eq:defft} and contour deformation, justified by 
the holomorphy of $ \hat g $ and the fact that $ e^{ i \zeta P } : \mathscr D^s \to 
\mathscr D^s $ for $ \Im \zeta \geq 0 $ (recall that $ P \geq 1 $): 
\[  e^{ - P }  \hat A  (g )  = e^{-P} \int e^{i t P } A e^{-i t P} \hat g ( t ) dt = 
\int_{\Gamma_R} e^{ i z P}  A e^{ -i ( z - i ) P } \hat g ( z -i) d z , \]
and the contour is given by $ \Gamma_R = \Gamma_R^+  + \Gamma_R^- +\gamma_R^+ +  \gamma_R^-  + I_R $, 
where $ I_R = [-R,R] $ with the positive orientation, and 
$\Gamma_R^\pm =  \pm [ R, \infty )  + i  $, $  \gamma_R^\pm = \pm R + i [ 0, 1 ] $, 
with matching orientations. Condition \eqref{eq:expf} then shows that
\[   \int_{\Gamma_R^\pm }  e^{ i z P}  A e^{ -i ( z - i ) P } \hat g ( z -i) d z , 
\int_{\gamma_R^\pm }  e^{ i z P}  A e^{ -i ( z - i ) P } \hat g ( z -i) d z  \to 0 , \ \ \ R \to \infty , \]
implying that
\[  e^{ - P }  \hat A  (g )  = \int_{\mathbb R} e^{ i t P} A e^{ - i t P } \hat g ( t - i ) dt e^{-P } . \]
Since $ \hat g ( \bullet - i ) = \widehat{ e^{ -\bullet } g ( \bullet ) } $, this gives \eqref{eq:ePAg}.
\end{proof}
With the lemma in place, the distributional point of view allows us to carry the calculations in 
\S \ref{s:localmat} in the case of $ f $ given by the Gaussians  \eqref{eq:defsigma}.
\end{proof}

To prove Theorem \ref{t:1} via Proposition \ref{p:L2D} below we start with some preliminaries about mapping properties of $ A $ and $ A^* $.
The assumption \eqref{e:assumeA} shows that the following operators are continuous:
\begin{equation}
\label{eq:mappA}     A :  \mathscr{D}^{\frac12} \to \mathcal{H}, \ \ \ A^* : \mathcal H\to\mathscr{D}^{-\frac12}.
\end{equation}
We also note that the assumption that $ [ P, A ] P^{-\frac12} : \mathcal H \to \mathcal H $
gives 
\begin{equation}
\label{eq:AD1}  \begin{split}  \| A \|_{ \mathscr D^{-\frac12}  \to \mathscr D^{-1} }&=\| A ^*\|_{ \mathscr D^1  \to \mathscr D^{\frac12} } \\
& = \| 
P^{\frac12} A^* P^{-1} \|_{ \mathcal H \to \mathcal H } \leq \| P^{-\frac12 } A^* \|_{ \mathcal H \to \mathcal H }
+ \| P^{\frac12} [ P^{-1} , A^* ]  \|_{ \mathcal H \to \mathcal H } \\
&=\| P^{-\frac12 } A^* \|_{ \mathcal H \to \mathcal H }
+ \| P^{-\frac12} [ P , A^* ] P^{-1} \|_{ \mathcal H \to \mathcal H }\\
&  \leq  \| AP^{-\frac12}  \|_{ \mathcal H \to \mathcal H } +  \| [ P , A ]P^{-\frac12}   \|_{ \mathcal H \to \mathcal H }.
\end{split} 
\end{equation}
since $ \| P^{-1} \|_{ \mathcal H \to \mathcal H } \leq 1$. Complex interpolation then implies that 
$$
\|A\|_{\mathscr{H}\to \mathscr{D}^{-\frac12}}\leq \| A \|^{1/2}_{ \mathscr D^{-1/2}  \to \mathscr D^{-1} }\| A \|_{ \mathscr D^{1/2}  \to \mathcal{H} }^{1/2}\leq C\|AP^{-\frac 12}\|_{\mathcal\to \mathcal{H}}+  \| [ P , A ]P^{-\frac12}   \|_{ \mathcal H \to \mathcal H }.
$$
We also have
\begin{equation}
\label{eq:AD2}  \begin{split}  \| A \|_{ \mathscr D^{\frac32}  \to \mathscr D^{1} } 
& = \| 
P AP^{-\frac32} \|_{ \mathcal H \to \mathcal H } \leq \| A P^{-\frac12}\|_{ \mathcal H \to \mathcal H }
+ \| [ P , A]P^{-\frac12}  \|_{ \mathcal H \to \mathcal H } 
\end{split} 
\end{equation}
and hence, using interpolation again $ A:\mathscr{D}\to \mathcal{H}$. 
Consequently, we also have the following uniform continuity statement: 
\begin{equation}
\label{eq:mappA1}    
e^{  i tP } A^* e^{i ( s -t  ) P } A e^{ - i s P } : \mathcal H \to \mathscr D^{-1} , \ \ \ 
e^{  i t P } A^* e^{i ( s - t ) P } A e^{ - i sP } : \mathscr D^1  \to \mathcal H  . 
\end{equation}

To prove Theorem \ref{t:1} we recall \eqref{eq:Daviesc}:
\begin{equation}
\label{eq:Daviesc1}
Y:=i( \beta^{-1} P+B)-\tfrac{1}{2}\sum_{A\in\mathcal{A}}\int \gamma(\omega)\hat{A}_f(\omega)^*\hat{A}_f(\omega)d\omega 
\end{equation} 
where $B$ is defined in~\eqref{eq:defB0}. In view of the results of \cite{dav}, Theorem \ref{t:1} follows from
\begin{prop}
\label{p:L2D}
Suppose that $ \gamma, f , b_j \in L^1 ( \mathbb R ) $, $ j =1,2 $, $ \gamma \geq 0 $, and that for some $ \delta > 0 $, 
\begin{equation}
\label{e:assumeAd} \sum_{A\in\mathcal{A}} \big(\|AP^{-\frac12}\|^2_{\mathcal{H}\to \mathcal{H}}+\|[P,A]P^{-\frac12}\|_{\mathcal{H}\to \mathcal{H}}^2+\|[P,[P,A]]P^{-1+\delta}\|_{\mathcal{H}\to \mathcal{H}}^2\big)<\infty,
\end{equation} 
Then $Y$ in \eqref{eq:Daviesc1} is the generator of a contraction semigroup on $ \mathcal H$.
\end{prop}

\begin{proof}[Proof of Proposition \ref{p:L2D}]
As in the proof of Lemma \ref{l:B2B}, we can write $ B $ and the last term on the right hand side of \eqref{eq:Daviesc1} as 
\begin{equation}
\label{eq:defR}  R := \int_{\mathbb R^2 } e^{i t P } A^* e^{ i ( s - t ) P } A e^{ - i s P } F ( t, s ) dt ds , \end{equation}
with
\[  F ( t, s ) = \tfrac12 b_1 ( - \tfrac12 ( t + s ) ) b_2 ( \tfrac12 ( t - s ) ) \in L^1 ( \mathbb R^2 ) , \]
and
\[ F ( t , s ) = \int_{\mathbb R } \gamma ( \omega ) e^{ i ( t - s ) \omega } f ( t ) f ( s ) d \omega \in L^1 ( \mathbb R^2 ) , \]
respectively (we assumed $ \gamma , f, b_j \in L^1 $). From \eqref{eq:mappA1} we see that
$ R : \mathcal H \to \mathscr D^{-1} $ is continuous and hence so is $ Y$. 

To apply the Hille-Yosida theorem (see \cite[Chapter 9, \S 1.2]{kato}), we need to show that $Y$ with the domain
$$
\mathcal{D}(Y):=\{ u\in \mathcal H \,:\, Yu\in \mathcal H \}. 
$$
(in view of the discussion above, we interpret $ Y u \in \mathscr D^{-1} \supset \mathcal H $) 
is closed and prove the estimate
\begin{equation}
\label{e:Yoside}
\|(Y-\lambda)^{-1}\|_{L^2\to L^2}\leq \lambda^{-1},\qquad \lambda>0.
\end{equation}

To show that $Y$ is closed we first consider $Y_0=Y|_{\mathscr D^1} $ and show that $\overline{Y}_0=Y$. For that, suppose $u_n\in \mathscr D^1 $ with 
$ u_n\overset{\mathcal H}{\to}u$, $ Y_0 u_n\overset{\mathcal H}{\to}w $. 
Since this implies that $Y_0u_n\overset{\mathscr D^{-1} }{\to} w$, and $Y$ is continuous from $ \mathcal H \to \mathscr D^{-1} $, we have $w=Yu$ (where the right hand side is understood to be in $ \mathscr D^{-1} $). In particular, $\overline{Y}_0\subset Y$. 

To show equality, let $u\in\mathcal{D}(Y)$. We will show that  
\begin{equation}
\label{e:useful}
\forall\,u\in\mathcal{D}(Y)\,\exists\, u_\epsilon\in \mathscr{D}^1  \qquad u_\epsilon\overset{\mathcal H}{\longrightarrow  }u,\quad Y_0u_\epsilon \overset{\mathcal H }{\longrightarrow }Yu,\quad \epsilon \to 0. 
\end{equation}
In other words, $ Y \subset \overline Y_0 $, and hence $ Y = \overline Y_0 $ is closed.

To prove this we need the following Lemma:
\begin{lemm}
\label{l:approximate}
Suppose that $R:\mathscr{S}\to \mathscr{S}$ satisfies
\begin{equation}
\label{eq:assR} 
\|P^{-\frac12} [P,R]P^{-\frac12}\|_{\mathcal H\to \mathcal H}+\| P^{-1+\delta} [P,[P,R]]P^{-1+\delta}\|_{\mathcal H\to \mathcal H}<\infty,
\end{equation}
for some $ \delta > 0 $. 
Then, for all $\chi_0,{\chi}_1 \in C_c^\infty(\mathbb{R})$ with ${\chi}_1 |_{[-1,1]}\equiv 1$  and $\supp (1-\chi_0)\cap \supp {\chi}_1 =\emptyset$, there is $C>0$ such that for $0<\epsilon<1$ there $A_\epsilon$ and $B_\epsilon$ such that
\begin{equation}
\label{e:approximator}
[\chi_0(\epsilon P),R]=A_\epsilon(1-{\chi}_1(\epsilon P))+\epsilon^{2\delta} B_\epsilon,\qquad 
\|B_\epsilon\|_{\mathcal H\to \mathcal H}+\|A_\epsilon\|_{\mathcal H\to \mathcal H}\leq C.
\end{equation}
\end{lemm}
\begin{proof}
We use the Helffer--Sj\"ostrand formula (see \cite[Theorem 8.1]{DiSj}):
\begin{equation}
\label{eq:HS} \chi_j ( \varepsilon P ) = \pi^{-1} \int_{\mathbb C } \bar \partial_z \tilde \chi_j ( z ) 
( \varepsilon P - z )^{-1} dm ( z ) , \end{equation}
where $ dm ( z ) $ is the Lebesgue measure on $ \mathbb C \simeq \mathbb R^2 $ and $ \tilde \chi_j \in C^\infty_{\rm{c}} ( \mathbb C ) $ satisfies $ \tilde \chi_j |_{\mathbb R } = \chi_j$, 
$ \bar \partial_z \tilde \chi_h ( z ) = \mathcal ( |\Im z |^\infty ) $. (That means that $ \tilde \chi_j $ is an 
{\em almost analytic} extension of $ \chi_j $.) We can assume that $ \widetilde \chi_0 = 1 $ on the support of $ \widetilde \chi_1 $. (See the construction in the beginning of \cite[Chapter 8]{DiSj}.) We then have 
the decomposition in \eqref{e:approximator} with 
\[ \begin{split}  A_\varepsilon & = [ \chi_0 ( \varepsilon P ) , R ] = 
 \pi^{-1} \varepsilon \int_{\mathbb C } \bar \partial_z \tilde \chi_0 ( z ) ( \varepsilon P - z )^{-1} [ R, P ]
 ( \varepsilon P - z )^{-1} dm ( z ) ,\\
 B_\varepsilon & =  \varepsilon^{-2\delta}  [ \chi_0 ( \varepsilon P ) , R ] \chi_1 ( \varepsilon P ) .
 \end{split} \]
To estimate $ \| A_\varepsilon \| $ we use the spectral theorem to see that for $ z \notin \mathbb R $, 
\begin{equation}
\label{eq:Pdel}     \| ( \varepsilon P - z )^{-1} P^{1-\delta} \|_{\mathcal H\to \mathcal H} \leq 
\sup_{ x \in [1,\infty )  } \frac{ x^{1-\delta} }{ |\varepsilon x - z| } \leq \frac{ | z|^{1-\delta}  } { \varepsilon^{1-\delta}  |\Im z |} . \end{equation}
This with $ \delta = \frac12 $ and the first bound in \eqref{eq:assR} show that $ \|A_\varepsilon \|_{\mathcal H \to \mathcal H } $ is uniformly bounded. 

To estimate the norm of $ B_\varepsilon $,
let $\chi_m\in C_{\rm{c}}^\infty( \mathbb R ) $ satisfy
\[ \supp (1-\chi_m)\cap \supp \chi_1=\emptyset, \ \ \ \supp\chi_m\cap \supp (1-\chi_0)=\emptyset. \] Then,
\begin{align*}
 B_\varepsilon & =  \varepsilon^{-2\delta}  [ \chi_0 ( \varepsilon P ) , R ] \chi_1 ( \varepsilon P ) 
= -  \varepsilon^{-2\delta}  (1- \chi_0 ( \varepsilon P )) R \chi_m(\varepsilon P) \chi_1 ( \varepsilon P ) 
  \\ &
  = -  \varepsilon^{-2\delta}  (1- \chi_0 ( \varepsilon P ) )[[R, \chi_m(\varepsilon P)], \chi_1 ( \varepsilon P )].
\end{align*}
Hence, using the Helffer Sj\"ostrand formula,
\begin{gather*}
B_\varepsilon= \pi^{-2}\varepsilon^{2-2\delta}(1-\chi_0(\varepsilon P))  \int_{\mathbb C^2 }  \bar \partial_z  \tilde{\chi} ( z )  \bar \partial_w \tilde \chi_m ( w )Q_{\varepsilon}(z,w)dm(z)dm(w),\\
Q_\varepsilon(z,w):=(\varepsilon P-w)^{-1}( \varepsilon P - z )^{-1}[[ R, P ],P] ( \varepsilon P - z )^{-1}(\varepsilon P-w)^{-1}.
\end{gather*}
Using~\eqref{eq:Pdel} with $\delta=0$ to estimate $(\varepsilon P-w)^{-1}:\mathcal{H}\to \mathcal{H}$,~\eqref{eq:Pdel} as stated to estimate $P^{1-\delta}(\varepsilon P-z)^{-1}:\mathcal{H}\to \mathcal{H}$, and the second hypothesis in~\eqref{eq:assR},  we see that $\|B_\varepsilon\|_{\mathcal{H}\to \mathcal{H}}\leq C$.
\end{proof}

To apply the lemma we need the following:
\begin{lemm}
\label{l:apply}
Suppose that \eqref{e:assumeAd} holds. Then for all $F\in L^1(\mathbb{R}^2)$, 
\eqref{eq:assR} holds for operators, $R$, of the form
$$
R=\sum_{A\in\mathcal{A}}\int_{\mathbb R^2 } e^{i t P } A^* e^{ i ( s - t ) P } A e^{ - i s P } F ( t, s ) dt ds
$$\end{lemm}
\begin{proof}
Since $P^{1-\delta}P^{-1+\delta'}:\mathcal{H}\to \mathcal{H}$ is bounded when $\delta'<\delta$, we may assume without loss of generality that $\delta \leq \frac{1}{2}$. 

To prove the lemma, we compute 
$$
[P,R]=\int_{\mathbb R^2 } e^{i t P } \big([P,A^*] e^{ i ( s - t ) P } A +A^* e^{ i ( s - t ) P } [P,A]\big)e^{ - i s P } F ( t, s ) dt ds.
$$
Hence, the bounds
$$
\sum_{A\in \mathcal{A}}\|[P,A]P^{-\frac12}\|^2+\|AP^{-\frac12}\|^2<\infty
$$
together with the fact that $P$ is self-adjoint imply the first estimate in~\eqref{eq:assR}.

Next, 
\begin{gather*}
[P,[P,R]]=\int_{\mathbb R^2 } e^{i t P } Q(t,s)e^{ - i s P } F ( t, s ) dt ds\\
Q(t,s):= [P,[P,A^*]] e^{ i ( s - t ) P } A +2[P,A^*]e^{i(s-t)P}[P,A]+A^* e^{ i ( s - t ) P } [P,[P,A]].
\end{gather*}
Hence, the bounds
$$
\sum_{A\in \mathcal{A}}\|[P,A]P^{-\frac12}\|^2+\|AP^{-\frac12}\|^2+\|[P,[P,A]]P^{-1+\delta}\|^2<\infty
$$
imply the second estimate in~\eqref{eq:assR}.
\end{proof}

By Lemma~\ref{l:apply}, \eqref{eq:assR} holds for $ Y$. Using the notation of Lemma~\ref{l:approximate}, 
we put $u_\epsilon=\chi_0 (\varepsilon P)u$.  Then, using \eqref{e:approximator} and
the assumptions in \eqref{e:useful}, 
\[   \begin{split} Y u_\varepsilon &  =  \chi_0 ( \varepsilon P ) Y u + [ Y , \chi_0 ( \varepsilon P ) ] u 
\\ & = Yu - ( 1 - \chi_0 ( \varepsilon P ) ) Y u + A_\varepsilon ( 1 - \chi_1 ( \varepsilon P ) ) u + 
\varepsilon^{2 \delta } B_\varepsilon  u \\
&   = Y u +  o(1)_{\mathcal H} + \mathcal O( \varepsilon^{2 \delta } ) \overset{\mathcal H}{\longrightarrow }0 , \ \ 
\varepsilon \to Yu.
\end{split}  \] 
(For any $ v \in \mathcal H $,  $ ( 1 - \chi_j ( \varepsilon P ) ) v \overset{\mathcal H}{\to }0 $ as 
$ \varepsilon \to 0 $.) This proves \eqref{e:useful} and completes the argument for the closedness of 
$ Y$. 

We now need to prove~\eqref{e:Yoside}. Observe that for $u\in\mathscr{D}^1$, the estimate \eqref{eq:AD1} shows that  $ A_f ( \omega ) \in \mathscr D^{\frac12} \subset \mathcal H $ and hence, 
$$
\Real \langle (Y-\lambda)u,u\rangle =-\lambda \|u\|_{\mathcal H}^2-\tfrac{1}{2}\int \gamma(\omega)\|A_f(\omega)u\|_{\mathcal H}^2d\omega\leq -\lambda \|u\|_{\mathcal H}^2. 
$$
For  $u\in\mathcal{D}(Y)$ and with $u_{\varepsilon}$ as in~\eqref{e:useful}. Then,
$$
\Real \langle (Y-\lambda)u,u\rangle =\lim_{\varepsilon \to 0} \Real \langle (Y-\lambda)u_{\varepsilon},u_{\varepsilon}\rangle\leq -\lambda \|u_{\varepsilon}\|_{\mathcal H}^2\to -\lambda \|u\|_{\mathcal H}^2. 
$$
Hence (using the same argument for $Y^*$ with its maximal domain), we have
\begin{equation}
\label{e:apriori}
\lambda \|u\|_{\mathcal H}\leq \|(Y-\lambda)u\|_{\mathcal H}, u\in\mathcal{D}(Y),\qquad \lambda \|u\|_{\mathcal H}\leq \|(Y^*-\lambda)u\|_{\mathcal H}, u\in\mathcal{D}(Y^*),
\end{equation}
Since~\eqref{e:apriori} implies $Y-\lambda$ is injective, and provides the estimate~\eqref{e:Yoside} if the inverse exists, it remains only to show that $Y$ is surjective. 
For this, suppose that $v\in \mathcal H$ such that 
$$
\langle (Y-\lambda)u,v\rangle =0,\qquad \forall u\in\mathscr{D}^{1} \subset \mathcal D ( Y ) . 
$$
Then, in $ \mathscr D^{-1} $,  $(Y^*-\lambda )v=0$, and hence $v\in \mathcal{D}(Y^*)$ and~\eqref{e:apriori} implies that $v=0$. 

We conclude that the inverse exists and \eqref{e:apriori} shows that \eqref{e:Yoside} holds, completing the proof of the proposition.
\end{proof}

\section{Proof of Theorem \ref{t:3}}
\label{s:pseudo}

In this section we show how the proof of Proposition \ref{p:L2D} can be modified to apply
to pseudodifferential operators quantising observables satisfying \eqref{eq:asp} and \eqref{eq:asA0}.
That will prove Theorem \ref{t:3}.

We define $ m = ( 1 + |x|^2 + |\xi|^2 )^{\frac12}   $ (a slightly different convention than in 
\S \ref{s:ex1} and the Appendix) and recall the following definitions: 
\[ \Psi ( m^r ) = \{ a^{\rm{w}} ( x, D ) : a \in S ( m^r ) \} , \ \ \   \Psi_{(k)} = \{ a^{\rm{w}} ( x, D ) : a \in S_{(k)} \}, \ \ k = 1, 2 , \]
where 
\[ \begin{gathered} 
a \in  S (m^r ) \Longleftrightarrow \partial^\alpha a = \mathcal O ( m^r  ) , \ \  | \alpha | \geq 0 ,\\
a \in  S_{ ( k ) }  \Longleftrightarrow \partial^\alpha a = \mathcal O ( 1 ) , \ \  | \alpha | \geq k , \ \ k = 0, 1, 2 . \end{gathered} \]
We note that $ S_{ (k ) } \subset S ( m^k ) $. Also all results are valid with only a finite (but large
depending on the dimension) number of derivatives needed. 

We want to investigate the structure of $ \chi ( \varepsilon P ) $ for $ \chi \in C_{\rm{c}}^\infty ( \mathbb R ) $,  and
\begin{equation}
\label{eq:defP}  P = p ( x, D ) , \ \  p \in S_{(2)}   , \ \   p \geq c m^2 . 
\end{equation}
One can show that $ P $ with the domain $ H (m^2 ) $ (see \cite[\S 8.3]{z12} for definitions; in this case it is particularly simple) is self-adjoint (see~\cite[Proposition A.2]{gaz5}). We assume in addition that $ P \geq 1 $ which can always be achieved by adding a constant to $ P $. We have mapping properties
\begin{equation}
\label{eq:Psr}    P^s : H ( m^r ) \to H( m^{r-2s} ), \ \  A : H (m^r ) \to H( m^{r-k} ),  \  A \in \Psi( m^k ) , \ \ \ 
  r, s \in \mathbb R . 
\end{equation}
The assumptions \eqref{eq:defP} and the fact that $ P $ is invertible implies that $ P^{-1}  \in \Psi (m^{-2} ) $ (see the Appendix). 

To follow the same strategy as in \S \ref{s:gen} it suffices the prove the following analogue of 
Lemmas \ref{l:approximate} and \ref{l:apply}:
\begin{prop}
\label{p:R1R2}
Suppose that $R_1 , R_2 \in \Psi_{(1)}$, $t\in\mathbb{R}$. Then, for 
\[ \chi,\chi_0\in C_c^\infty(\mathbb{R}), \ \ \ \supp \chi_0\cap \supp(1-\chi)=\emptyset, \]
\begin{equation}
\label{e:approximator1}
[\chi(\epsilon P), R_1 e^{itP}R_2 ]=A_\epsilon(1-{\chi}_0(\epsilon P))+\epsilon  B_\epsilon,\ \ \ 
\|B_\epsilon\|_{L^2\to L^2}+\|A_\epsilon\|_{L^2 \to L^2 }\leq C, 
\end{equation}
where $ C $ is independent of $ \varepsilon $. 
\end{prop}

The proof is based on the following 
\begin{lemm}
\label{l:funct1}
Suppose that $ \chi \in C^\infty_{\rm{c}} ( \mathbb R )$ and $ P $ satisfies the assumptions above.  Then
\begin{equation}
\label{eq:funct1}     \chi ( \epsilon p^{\rm{w}} ( x, D  ) ) =  (\epsilon p)^* \chi ^{\rm{w}} ( x,  D ) + \epsilon q_{\varepsilon}^{\rm{w}} ( x,  D) , \ \ \  
q_\varepsilon \in  S ( m^{-2} ) , \end{equation} 
uniformly in $ \varepsilon$, that with $ |\partial^\alpha q_\varepsilon | \leq C_\alpha m^{-2}  $, with constants independent of $ \varepsilon $. 
\end{lemm}

\noindent
{\bf Remark.} We obtain a stronger result about $ q_\varepsilon $ but we do not stress it as 
\eqref{eq:funct1} is sufficient for our purposes.

\begin{proof}
In what follows we assume that $ |z| \leq C $. 
We first show that for for all $M\geq 0$ there is $N>0$ such that for $\{\ell_j\}_{j=1}^M$ linear, 
\begin{equation}
\label{e:beals1}
\|\ad_{\ell_1^{\rm{w}}}\ad_{\ell_2^{\rm{w}}}\dots \ad_{\ell_M^{\rm{w}}}(\epsilon P-z)^{-1}\|_{L^2\to L^2}\leq C_M |\Im z|^{-N_M}\begin{cases}\epsilon&M\geq 2\\ \epsilon^{\frac{1}{2}M}&0\leq M\leq 1,
\end{cases}
\end{equation}
Hence for $ \Im z \neq 0 $, Beals's Lemma, or rather \cite[Theorem 8.1]{z12}, 
\begin{equation}
\label{eq:res2a}   (\epsilon P-z)^{-1} = a^{\rm w}_\varepsilon ( z, x, D ) , \ \ \ 
| \partial^\alpha a_\varepsilon | \leq C  |\Im z |^{-N_\alpha } \varepsilon^{\min(1,|\alpha|/2) } . 
\end{equation}
(Here and elsewhere $ N_\alpha $ denotes a constant depending on $ \alpha $ which may be different in different estimates.)

To prove \eqref{e:beals1}, we first observe that $\ad_{\ell_1^{\rm{w}}}P\in \Psi_{(1)}$. Since 
\[ \ad_{\ell^{\rm{w}}} (\epsilon P-z)^{-1} = - \varepsilon (\epsilon P-z)^{-1}  (\ad_{\ell^{\rm{w}}} P) P^{-\frac12} P^{\frac12} (\epsilon P-z)^{-1} , \]
and 
\[ \|   ( \ad_{\ell^{\rm{w}}} P ) P^{-\frac12} \|^2_{L^2 \to L^2 } = 
\|  ( \ad_{ \ell^{\rm{w}}} P ) P^{-1} ( \ad_{ \ell^{\rm{w}}} P )^* \|_{L^2 \to L^2 } \leq C , \]
(since $ P^{-1} \in \Psi (m^{-2} ) $ and $ \ad_{ \ell^{\rm{w}}} P \in \Psi_{(1)} \subset \Psi ( m ) $) and
$ \|  ( \varepsilon P - z )^{-1} P^{-\frac12} \| \leq  C \varepsilon^{\frac12} |\Im z |^{-1} $,  we have
\begin{equation}
\label{eq:bound1}  \| \ad_{\ell^{\rm{w}}} (\epsilon P-z)^{-1}  \|_{ L^2 \to L^2 } \leq C \varepsilon^{\frac12} |\Im z|^{-2} . \end{equation}

For $ M \geq 2 $ we argue as in the proof of \cite[Proposition 8.6]{DiSj} but for our class of operators. For that we note that 
$$
\ad_{\ell_1^{\rm{w}}}\dots \ad_{\ell_M^{\rm{w}}}(\epsilon P-z)^{-1}
$$
is a sum of terms of the form 
$$
\epsilon ^L (\epsilon P-z)^{-1} ( \ad_{\ell_1^{\rm{w}}}\dots\ad_{\ell_{j_1}^{\rm{w}} }  P ) (\epsilon P-z)^{-1}\dots ( \ad_{\ell_{j_{L-1}+1}^{\rm{w}}}\dots\ad_{\ell_{M}^{\rm{w}} }  P ) (\epsilon P-z)^{-1},
$$
where $1\leq L\leq M$ and and $1\leq j_1<j_2\dots<j_{L-1}<M$. The bound is then 
seen from considering the two extreme cases:
\[  \epsilon^M  (\epsilon P-z)^{-1} ( \ad_{\ell_1^{\rm{w}} }  P) ( \epsilon P-z)^{-1} 
 ( \ad_{\ell_2^{\rm{w}} }  P) ( \epsilon P-z)^{-1}  \cdots   ( \ad_{\ell_M^{\rm{w}} }  P) 
 ( \varepsilon P - z )^{-1} \]
 and
 \[ \varepsilon (\epsilon P-z)^{-1} ( \ad_{\ell_1^{\rm{w}} }  \cdots \ad_{\ell_M^{\rm{w}} } P) 
 ( \varepsilon P - z )^{-1} . \]
 In the first case we proceed as in the proof of \eqref{eq:bound1} to obtain the bound
 $ \varepsilon^{M/2 } |\Im z |^{-M-1}  = \mathcal O ( \varepsilon ) |\Im z |^{-M-1} $. In the second we use the fact that for $ M \geq 2 $, $  \ad_{\ell_1^{\rm{w}} }  \cdots \ad_{\ell_M^{\rm{w}} } P \in \Psi_{(0)} $, so the bound becomes 
 $ \mathcal O ( \varepsilon ) |\Im z |^{-2} $.

We write
$$
 [ ( \varepsilon p - z )^{-1} ]^{\rm{w}} ( \epsilon P-z)= I + \varepsilon E ( z ) ,
$$
where $ E ( z ) = e^{\rm{w}} ( z, x, D ) $ with 
\[ e ( z, x, \xi ) =   \int_0^1 ( 1 - t ) 
 e^{ i t A ( D ) } ( i  A ( D ) )^{2} ( \varepsilon^{-1} 
( \varepsilon p ( \rho_1 ) - z ) ( \varepsilon p ( \rho_2 )  - z )^{-1} )|_{\rho_1=\rho_2= (x,\xi) } \]
and $ A ( D ) := -  \frac12 \sigma ( D_{\rho_1} , D_{\rho_2 } ) $. 
The terms to which $ e^{ i t A ( D ) } $ is applied are,  schematically,  of the form 
\[   D^2 p ( \rho_1 ) \left( \frac{\varepsilon D^2 p( \rho_2 )}{  ( \varepsilon p ( \rho_2 ) - z )^{2 } }  +  \frac{\varepsilon^2 Dp ( \rho_2 ) Dp( \rho_2  ) }{ ( \varepsilon p ( \rho_2 ) - z )^3 } \right) . \]
Using that $p\geq cm^2$, for $ |z| < 1 $ and $\varepsilon>0$ small enough, we have
\[     |  ( \varepsilon p ( \rho_2 ) - z ) |^{-1}  \leq  |\Im z |^{-1} \min \left( 1 , 2 \varepsilon^{-1} c^{-1}m ( \rho_2 ) ^{-2}  \right). \]
Thus, for $ | \Im z | > 0 $, we have
\[  \partial^\alpha_{ \rho_1, \rho_2 } ( i  A ( D ) )^{2}  
( \varepsilon^{-1} ( \varepsilon p ( \rho_1 ) - z ) ( \varepsilon p ( \rho_2 )  - z )^{-1} ) = \mathcal O ( 
| \Im z |^{-N_\alpha } 
m ( \rho_2 )^{ - 2} )  . 
\] 
Since in the estimates $ e^{ i t  A( D ) } : S ( m_1( \rho_1 ) m_2 ( \rho_2 ) ) 
\to S ( m_1 ( \rho_1 ) m_2 ( \rho_2 ) ) $ only finitely many derivatives are used, we conclude that
\[  \partial^\alpha e ( z , x , \xi ) \leq  C |\Im z|^{-N_\alpha }  m ( x, \xi )^{-2} . \]
Combining this with \eqref{eq:res2a} we obtain
\[  ( \varepsilon P - z )^{-1} = [( \epsilon p(x,\xi)-z)^{-1}]^{\rm{w}}+ \varepsilon {R}(z), 
\]
where  $ R ( z ) = E ( z) ( \varepsilon P - z )^{-1} = r^{\rm w} ( z, x , D ) $ and
\[ \partial^\alpha  r ( z, x , \xi ) = \mathcal O ( | \Im z |^{-N_\alpha } 
m ( x, \xi)^{-2} ) . \]
We now apply the Helffer--Sj\"ostrand formula \eqref{eq:HS} (see the proof of Lemma \ref{l:approximate}) which gives \eqref{eq:funct1}. 
\end{proof}

\begin{proof}[Proof of Proposition \ref{p:R1R2}]
We put 
\[ A_\varepsilon := [\chi(\epsilon P), R_1 e^{itP}R_2] , \ \ \ B_\varepsilon = 
\varepsilon^{-1} [\chi(\epsilon P), R_1 e^{itP}R_2] \chi_0 ( \varepsilon P ) , \]
so that
\[ \begin{split} A_\varepsilon & =  [\chi(\epsilon P),R_1 ]e^{itP}R_2 +R_1 e^{itP}[\chi(\epsilon P),R_2] \\
& =  [\chi(\epsilon P),R_1  ] P^{\frac12} e^{itP} P^{-\frac12 } R_2 + 
R_1 P^{-\frac12} e^{itP} P^{\frac12} [\chi(\epsilon P), R_2] 
\end{split} \]
In view of \eqref{eq:Psr} and \eqref{eq:funct1}, to show that $ A_\varepsilon : L^2 \to L^2 $ is uniformly bounded, it is enough to show that
\begin{equation}
\label{eq:est0}  \| [ ( \varepsilon  p)^*\chi^{\rm{w}}  , R^* ] P [ ( \varepsilon  p)^*\chi^{\rm{w}}  , R ] \|_{L^2 \to L^2 } \leq C , \ \ R \in \Psi_{(1)}  \end{equation}
with $ C $ independent of $ \varepsilon $. Since $ m  \partial^\alpha ( \chi ( \varepsilon p ) ) = 
 \mathcal O( \varepsilon p ) ( \partial^\alpha \chi ) ( \varepsilon p )  
\mathcal O ( 1 ) $ for $ |\alpha | = 1 $, this follows from the composition formula as in 
\cite[Proposition A.1]{gaz5}. 

To analyse $ B_\varepsilon $, let $\chi_1\in C_{\rm{c}} ^\infty(\mathbb{R})$ with $\supp \chi_0\cap \supp (1-\chi_1)=\emptyset$ and $\supp \chi_1\cap \supp (1-\chi ) =\emptyset$. Then, 
\begin{equation} 
\label{eq:Beps} \begin{split}
  B_\varepsilon 
& = \varepsilon^{-1}  [\chi(\epsilon P),R_1 ]e^{itP}R_2\chi_1(\epsilon P)\chi_0(\epsilon P)+ \varepsilon^{-1} R_1 e^{itP}[\chi(\epsilon P),R_2]\chi_0(\epsilon P)\\
&  =  \varepsilon^{-1} [\chi(\epsilon P),R_1 ]  e^{itP} [R_2,\chi_1(\epsilon P)]\chi_0(\epsilon P) \\
& \ \ \ \ +  \varepsilon^{-1} [\chi(\epsilon P),R_1 ]\chi_1(\epsilon P) P^{\frac12} e^{itP} P^{-\frac12} R_2\chi_0(\epsilon P)  \\
& \ \ \ \ \ \ + \varepsilon^{-1} R_1 P^{-\frac12} e^{itP} P^{\frac12} [\chi(\epsilon P),R_2]\chi_0(\epsilon P)\\
&=:I +I\!I +I\!I\!I.
\end{split} \end{equation}
To prove the lemma, we show that $I$, $I\!I$, and $I\!I\!I$ are $O(1)_{L^2\to L^2}$. 

We start by estimating $I$.  To do this, we claim that 
\begin{equation}
\label{eq:est1}   \| [ \chi(\epsilon P), R ] \|_{L^2 \to L^2 }
= \mathcal O ( \varepsilon^{\frac12} ) . 
\end{equation}
The form of $I$ then implies $I=O(1)_{L^2\to L^2}$. 
To prove~\eqref{eq:est1}, we use \eqref{eq:Psr} and \eqref{eq:funct1} to replace the cut-off operator with $\big((\varepsilon p)^*\chi\big)^{\rm w}(x,D)$ so that it is enough to show
\begin{equation}
\label{eq:est1b}   \| [ ( \varepsilon  p)^*\chi^{\rm{w}} , R ] \|_{L^2 \to L^2 }
= \mathcal O ( \varepsilon^{\frac12} ) . 
\end{equation}
This follows by arguing as in the proof of \eqref{eq:est0}.

Next, to estimate $I\!I$, we claim that 
\begin{equation}
\label{eq:est2}
\|\varepsilon^{-1} [\chi(\epsilon P),R_1 ]\chi_1(\epsilon P) P^{\frac12} \|_{L^2\to L^2}=O(1).
\end{equation}
Since $ \| P^{-\frac12}R_2 \|^2 = \| R_2^* P^{-1}  R_2 \| =O(1)_{L^2\to L^2}$,  this implies $I\!I=O(1)_{L^2\to L^2}.$

To prove~\eqref{eq:est2} we put $ \tilde{\chi}_1(x) : =x\chi_1(x) \in C_{\rm{c}} ^\infty ( \mathbb R ) $ 
and observe that
\begin{equation}
\label{eq:est2b}
\begin{aligned}
&\|\varepsilon^{-1} [\chi(\epsilon P),R_1 ]\chi_1(\epsilon P) P^{\frac12} \|_{L^2\to L^2}^2\\
&=\|\varepsilon^{-2} [\chi(\epsilon P),R_1 ]\chi_1(\epsilon P) P\chi_1(\epsilon P)[\chi(\epsilon P),R_1] \|_{L^2\to L^2}\\
&=\|\varepsilon^{-3} [\chi(\epsilon P),R_1 ]\tilde{\chi}_1(\epsilon P) \chi_1(\epsilon P)[\chi(\epsilon P),R_1] \|_{L^2\to L^2}\\
&\leq \epsilon^{-3}\|[\chi(\epsilon P),R_1 ]\tilde{\chi}_1(\epsilon P)\|_{L^2\to L^2}\|[\chi(\epsilon P),R_1]\chi_1(\epsilon P) \|_{L^2\to L^2}.
\end{aligned}
\end{equation}
Now, we claim that for $\psi\in C_c^\infty$ with $\supp \psi\cap \supp (1-\chi)=\emptyset$, 
\begin{equation} \label{eq:est2c}
\|[\chi(\epsilon P),R_1 ]\psi(\epsilon P) \|_{L^2\to L^2}= O(\epsilon^{3/2}).
\end{equation}
The estimate~\eqref{eq:est2} then follows by using~\eqref{eq:est2c} with $\psi=\tilde{\chi}_1$ and $\psi=\chi_1$ in~\eqref{eq:est2b}. 

To prove~\eqref{eq:est2c} we again use \eqref{eq:Psr} and \eqref{eq:funct1} to replace the cut-off operator with $\big((\varepsilon p)^*\chi\big)^{\rm w}(x,D)$.  That is, we need to show
\begin{equation} \label{eq:est2c}
\|[(\epsilon p)^*\chi^{\rm w},R_1 ](\epsilon p)^*\psi^{\rm w}\|_{L^2\to L^2}= O(\epsilon^{3/2}).
\end{equation}
The symbolic calculus (see \cite[Proposition A.1]{gaz5}) gives
\begin{equation}
\label{e:cat1}
[(\epsilon p)^*\chi^{\rm w},R_1 ]= \epsilon^{1/2} \big(\epsilon ^{1/2}\{p, r_1\}\chi'(\epsilon  p)\big)^{\rm w} + \epsilon b_1^{\rm w},
\end{equation}
where $b_1\in S_{(0)}$.  Since $\{p, r_1\}\in S(m)$,  and $p\geq m^2 $, 
$$
\epsilon^{1/2} \{p, r_1\}p\chi'(\epsilon  p)\in S_{(0)}.
$$
Therefore, using \cite[Proposition A.1]{gaz5} again and the fact that $\supp \chi'\cap \supp \psi=\emptyset$ in the first equality
\begin{equation}
\label{e:cat2}
\begin{gathered}
\epsilon^{1/2} \big(\epsilon ^{1/2}\{p, r_1\}\chi'(\epsilon  p)\big)^{\rm w} (\epsilon p)^*\psi^{\rm w}= \epsilon ^{3/2} b_2^{\rm w},\qquad
\epsilon b_1^{\rm w}(\epsilon p)^*\psi^{\rm w}= \epsilon^{3/2}b_3^{\rm w},
\end{gathered}
\end{equation}
where $b_2,b_3\in S_{(0)}$.  Putting~\eqref{e:cat1} and~\eqref{e:cat2} together implies~\eqref{eq:est2c} and hence completes the proof of the bound on $I\!I$.

Finally, to estimate $I\!I\!I$ we observe that 
$$
\|P^{\frac12} [\chi(\epsilon P),R_2]\chi_0(\epsilon P)\|_{L^2\to L^2}^2= \|\chi_0(\epsilon P)[\chi(\epsilon P),R_2] P[\chi(\epsilon P),R_2]\chi_0(\epsilon P)\|_{L^2\to L^2}
$$
and 
\begin{align*}
&\chi_0(\epsilon P)[\chi(\epsilon P),R_2] P[\chi(\epsilon P),R_2]\chi_0(\epsilon P)\\
&=\chi_0(\epsilon P)[\chi(\epsilon P),R_2][\chi(\epsilon P),R_2]P\chi_0(\epsilon P)+ [\chi_0(\epsilon P), R_2][P,[\chi(\epsilon P), R_2]]\chi_0(\epsilon P)\\
&=\chi_0(\epsilon P)[\chi(\epsilon P),R_2][\chi(\epsilon P),R_2]P\chi_0(\epsilon P)- [\chi_0(\epsilon P), R_2][\chi(\epsilon P),R_3]\chi_0(\epsilon P),
\end{align*}
where $R_3:= [R_2,P]\in \Psi_{(1)}$. Arguing as in the proof of the estimate on $I\!I$, we obtain  
\begin{gather*}
\|[\chi(\epsilon P),R_2]P\chi_0(\epsilon P)\|_{L^2\to L^2}=O(\epsilon ^{1/2}),\quad \|\chi_0(\epsilon P)[\chi(\epsilon P),R_2]\|_{L^2\to L^2}=O(\epsilon^{3/2}),\\
\|[\chi(\epsilon P),R_3]\chi_0(\epsilon P)\|_{L^2\to L^2}=O(\epsilon^{3/2}).
\end{gather*}
Therefore, using~\eqref{eq:est1} (or more precisely its analog with $\chi$ replaced by $\chi_0$) completes the proof that $I\!I\!I=O(1)_{L^2\to L^2}$ and hence of the lemma.
\end{proof}

\section{A new Egorov theorem and the proof of Theorem \ref{t:4}}

In this section, we prove the version of Egorov's theorem needed for
Theorem 4.  For standard finite-time versions and references see
\cite[Chapter 11]{z12}, and for some recent advances, \cite{bon} and
\cite{pro}.  The point here is that, for the symbol classes used in this
paper, the conjugated operator remains pseudodifferential for all
times, with symbol seminorms allowed to grow exponentially.

Recalling that $ a \in S_{(k)} $ if $ \partial^\alpha a = \mathcal O ( 1 ) $ for $ |\alpha| \geq k $, we assume that $p\in S_{(2)}$ is real valued and define $P:=p^{\rm w}$. We will show that for $b\in S_{(1)} $, and $t\in\mathbb{R}$, 
$
e^{itP}b^{\rm w}e^{-itP}
$
is a pseudodifferential operator and obtain quantitative control on the derivatives of its symbol. More precisely, we prove the following:
\begin{theo}
\label{t:egorov}
Suppose that $b\in S_{(1)}$. Then for all $k\geq 0$ there are $C_k>0$ such that for all $t\in\mathbb{R}$ there is $r_t\in S_{(0)}$ satisfying
$$
e^{itP}b^{\rm w}e^{-itP}=(b\circ\varphi_t)^{\rm w}+r_t^{\rm w}
$$
and
$$
\|r_t\|_{C^k}\leq C_ke^{C_k|t|}. 
$$
\end{theo}
\begin{rem}
It is possible to prove a version of Theorem~\ref{t:egorov} for $b\in S_{(\ell)}$ for any $\ell \geq 0$, using the calculus in $\Psi_{(\ell)}$ to reduce to the case of $S_{(0)}$. However, since we do not need it here, we do not pursue this generalisation.
\end{rem}
Theorem~\ref{t:egorov} allows us to conclude that the jump operators $\hat A_f ( \omega )$ are pseudodifferential when $f$ is super-exponentially decaying and $A=a^{\rm w}$ with $a\in S_{(1)}$. 
The structure of the coherent term is more complicated as $ b_1 ( t ) $ is {\em not} super-exponentially decreasing.

To prove Theorem~\ref{t:egorov},  we start with an estimate on the flow for $p\in S_{(2)}$. For this, we follow~\cite[Lemma 11.1]{z12}.
\begin{lemm}\label{l:flow}
Let $p\in S_{(2)}$ be real valued and $\varphi_t:=\exp(tH_p)$. There exist $\Lambda>0$ and $C_{\alpha}>0$,  such that for all $b\in S_{(1)}$ and all $t\in\mathbb{R}$, 
$$
|\partial^\alpha ( b\circ\varphi_t) |\leq C_\alpha e^{\Lambda|\alpha||t|}\sum_{1\leq |\beta|\leq |\alpha|}\|\partial^\beta b\|_{L^\infty}, \quad |\alpha|\geq 1, \ \ \alpha \in\mathbb{N}^{2d}.
$$
\end{lemm}
\begin{proof}
The lemma follows from showing that 
\begin{equation}
\label{e:flowEstimate}
|\partial^\alpha \varphi_t(x,\xi)|\leq C_{\alpha}e^{\Lambda|\alpha| |t|},\qquad |\alpha|\geq 1.
\end{equation}
The flow $\rho(t)=\varphi_t((x_0,\xi_0))$ is defined by 
\begin{equation}
\label{e:hamiltonian}
\dot\rho (t)=H_p(\rho(t)),\qquad \rho(0)=(x_0,\xi_0).
\end{equation}
Differentiating~\eqref{e:hamiltonian} once, we obtain for $|\alpha|=1$, 
$$
\frac{d}{dt}(\partial^\alpha \varphi_t)=\partial H_p(\rho(t)) \partial^\alpha \varphi_t.
$$
Therefore, since $\partial H_p\in L^\infty$, we obtain
$$
|\partial ^\alpha \varphi_t|\leq e^{\|\partial H_p\|_{\infty}|t|},\qquad |\alpha|=1, 
$$
which gives \eqref{e:flowEstimate} with $ \Lambda \geq \|\partial H_p\|_{\infty} $. 

Let $\ell\geq 2$ and suppose by induction that~\eqref{e:flowEstimate} holds for all $1\leq |\alpha|\leq \ell-1$.  Then for $ |\alpha | = \ell$,  differentiating the equation~\eqref{e:hamiltonian}, we obtain
\begin{equation}
\label{e:manyDerivatives}
\frac{d}{dt}(\partial^\alpha \varphi_t)=\partial H_p(\rho(t)) \partial^\alpha \varphi_t+\gamma(t),
\end{equation}
where $\gamma(t)$ is a sum of terms of the form
$$g\, \partial^{\alpha_1}\varphi_t\dots \partial^{\alpha_k}\varphi_t,
$$
where $\alpha_1+\dots+\alpha_k=\alpha$ and $0<|\alpha_j|<|\alpha|=\ell$ and $g$ is bounded. The induction hypothesis implies
$$
|\gamma(t)|\leq Ce^{\Lambda |\alpha||t|},
$$
and using this in~\eqref{e:manyDerivatives} gives
$$
\frac{d}{dt}|\partial^\alpha\varphi_t|^2\leq |\partial^\alpha\varphi_t|^2\|\partial H_p\|_{\infty}+Ce^{\Lambda |\alpha||t|}|\partial^\alpha \varphi_t|.
$$
This implies~\eqref{e:flowEstimate} after taking $\Lambda >\|\partial H_p\|_{\infty}$ and an application of Gr\"onwall's inequality.
\end{proof}

Next, we give a preliminary decomposition of the conjugated operator.
\begin{lemm}
\label{l:initialEgorov}
There is $N>0$ such that for all $a\in S_{(1)}$ with
\begin{equation}
\label{e:At}
A_t:= e^{itP}a^{\rm w}e^{-itP},
\end{equation}
and all $t\in\mathbb{R}$, there is $q_t\in S_{(0)}$ satisfying
\begin{equation}
\label{eq:defQ}
A_t=(a\circ\varphi_{t})^{\rm w}-Q(a,t),\qquad Q(a;t):=i\int_0^t e^{i(t-s)P}q_s^{\rm w}e^{i(s-t)P}ds,
\end{equation}
and 
\begin{equation}
\label{e:qEstimates}
|\partial^\alpha q_t|\leq C_\alpha e^{\Lambda (|\alpha| +N)|t|}\sum_{1\leq |\beta|\leq |\alpha|+N}\|\partial^\beta a\|_{\infty}. 
\end{equation}
\end{lemm}
\begin{proof}
Observe that
\begin{equation}
\label{eq:Dta} \begin{split}
D_t\big( e^{-itP}(a\circ \varphi_t)^{\rm w} e^{itP} \big)
&= e^{-itP}\big( -[P , (a\circ \varphi_t)^{\rm w}] - i (H_p(a\circ \varphi_{t}))^{\rm w} \big) e^{itP} \\
&= : e^{-itP} q_t^{\rm w} e^{itP}.
\end{split} \end{equation}
This definition of $ q_t$, \cite[Lemma A.1]{gaz5} and Lemma \ref{l:flow} 
show that there is $N>0$ such that for all $ \alpha $, 
\[ \begin{split} 
|\partial^\alpha q_t| & \leq C_\alpha\sum_{2\leq|\beta|\leq |\alpha|+N}\|\partial^\beta (a\circ\varphi_{t})\|_{\infty}  
 \leq C_\alpha e^{\Lambda(|\alpha|+N)|t|}\sum_{1\leq |\beta|\leq |\alpha|+N}\|\partial^\beta a\|_{\infty}.
\end{split}  \]
In particular, $ q_t \in S_{(0)} $. Integrating \eqref{eq:Dta}  in $t$ gives
\begin{align*}
e^{itP}a^{\rm w}e^{-itP}-(a\circ \varphi_t)^{\rm w}
&=-i\int_0^t D_s\big(e^{i(t-s)P}(a\circ\varphi_{s})^{\rm w}e^{i(s-t)P}\big)ds\\
&=-i\int_0^t e^{i(t-s)P}q_s^{\rm w}e^{i(s-t)P}ds
\end{align*}
which is \eqref{eq:defQ}. 
\end{proof}

Before proceeding to the full case of $S_{(1)}$, we prove the Egorov theorem for $S_{(0)}$. 
\begin{lemm}
\label{l:zeroEgorov}
Let $p\in S_{(2)}$ be real valued. Then, for $k\geq 0$, there are $C_k,N_k>0$ such that for all $t\in\mathbb{R}$ and $a\in S_{(0)}$, there is $a_t\in S_{(0)}$ satisfying
$$
A_t:=e^{itP}a^{\rm w}e^{-itP}=a_t^{\rm w},\qquad 
\| a_t\|_{C^k}\leq C_{k} e^{C_{k}|t|}\|a\|_{C^{N_k}}. 
$$
\end{lemm}
\begin{proof}
For this, we employ Beals's characterisation of pseudodifferential operators~\cite[Theorem 8.12]{z12}. In particular, it is enough to show that for any $M\geq 0$ there is $N_M>0$ such that for any fixed linear functions $L_j(x,\xi)$, $j=1,\dots M$, with $\|\nabla L_j\|\leq 1$, there are $C_M>0$  such that 
\begin{equation}
\label{e:beals}
\|\ad_{L_1^{\rm w}}\dots \ad_{L_M^{\rm w}}A_t\|_{L^2\to L^2}\leq C_M e^{C_M|t|}\|a\|_{C^{N_M}}.
\end{equation}

\noindent Step 1: Boundedness of commutators with operators of linear growth.  \ 

In the notation of \eqref{eq:defQ}, and with $ B_t := e^{itP}b^{\rm w}e^{-itP}$, $ b \in S_{(1)} $, 
\begin{equation}
\label{e:singleCommutatorForm}
e^{-itP}[b^{\rm w},A_t]e^{itP}=[B_{-t},a^{\rm w}]= [(b\circ\varphi_{-t})^{\rm w}-Q(b;-t),a^{\rm w}],
\end{equation}
where 
$$
Q(b,-t):=i\int_0^{-t}e^{-i(t+s)P}q_s^{\rm w}e^{i(s+t)P}ds
$$
and $q_s$ satisfies~\eqref{e:qEstimates} with $a$ replaced by $b$. 
Then unitarity of $e^{itP}$ and~\cite[Theorem 4.23]{z12} show that there is $N>0$ such that 
$$
\|Q(b;-t)\|_{L^2\to L^2}\leq C\sum_{0\leq |\alpha|\leq N}\int_{\min(0,-t)}^{\max(0,-t)}\|\partial^\alpha q_s\|_{\infty}\leq C|t|e^{2N\Lambda|t|}\sum_{1\leq |\beta|\leq 2N} \|\partial^\beta b\|_{\infty}
$$
and hence (absorbing the $|t|$ into the exponential factor)
we conclude that for $ N_0 $ large enough, 
\begin{equation}
\label{e:commuteBound}
\|Q(b;-t)\|_{L^2\to L^2}\leq Ce^{\Lambda M_0|t|} \sum_{1\leq |\beta|\leq N_0}\|\partial^\beta b\|_{\infty}.
\end{equation}
Lemma~\ref{l:flow},~\cite[Lemma A.1]{gaz5}, and~\cite[Theorem 4.23]{z12}, combined with~\eqref{e:singleCommutatorForm} and~\eqref{e:commuteBound} then imply that there are $M_0,N_0>0$ such that for $b\in S_{(1)}$,
\begin{equation}
\label{e:singleCommutator}
\|[b^{\rm w},A_t]\|_{L^2\to L^2}\leq Ce^{ \Lambda M_0|t|}\sum_{1\leq |\beta|\leq N_0}\|\partial^\beta b\|_{\infty}\sum_{0\leq |\beta|\leq N_0}\|\partial^\beta a\|_{\infty}.
\end{equation}

\noindent Step 2: Commutators with error terms. \ 

In the inductive process, we also require the estimate: for $b_1, b\in S_{(1)}$,
\begin{equation}
\label{e:commuteBound2}
\|[(b_1\circ\varphi_{-t})^{\rm w},Q(b;-t)]\|_{L^2\to L^2}\leq Ce^{3\Lambda M_0|t|}\sum_{1\leq |\beta|\leq N_0}\|\partial^\beta b_1\|_{\infty} \sum_{1\leq |\beta|\leq N_0}\|\partial^\beta b\|_{\infty} 
\end{equation}
To prove~\eqref{e:commuteBound2}, observe that for $\tilde{b}_1\in S_{(1)}$, 
\begin{align*}
[Q(b;-t),\tilde{b}_1^{\rm w}]&=-i\int_{-t}^{0}[e^{i(-t-s)P}q_s^{\rm w}e^{i(s+t)P},\tilde{b}_1^{\rm w}]ds.
\end{align*}
Now, by~\eqref{e:qEstimates} and~\eqref{e:singleCommutator} (with $ a $ replaced by 
$ q_s $ and $ t $ by $ t + s $)
\begin{align*}
&\|[e^{i(-t-s)P}q_s^{\rm w}e^{i(s+t)P},\tilde{b}_1^{\rm w}]\|_{L^2\to L^2}\\
&\leq Ce^{M_0\Lambda |t+s|}\sum_{1\leq |\beta|\leq N_0}\|\partial^\beta \tilde{b}_1\|_{\infty}\sum_{0\leq |\alpha|\leq N_0}\|\partial^\alpha q_s\|_{\infty}\\
&\leq Ce^{\Lambda(M_0|t+s|+(N_0+N)|s|)}\sum_{1\leq |\beta|\leq N_0}\|\partial^\beta \tilde{b}_1\|_{\infty}\sum_{1\leq |\beta|\leq N_0+N}\|\partial^\beta b\|_{\infty}\\
&\leq Ce^{2\Lambda M_0|t|}\sum_{1\leq |\beta|\leq N_0}\|\partial^\beta \tilde{b}_1\|_{\infty}\sum_{1\leq |\beta|\leq N_0+N}\|\partial^\beta b\|_{\infty},
\end{align*}
where the last line follows if $M_0$ is large enough.

Now, setting $\tilde b_1=b_1\circ\varphi_{-t}$, using Lemma~\ref{l:flow}, enlarging $N_0$, choosing $M_0$ sufficiently large, and absorbing the factor $|t|$ into the exponential, we obtain~\eqref{e:commuteBound2}.

\noindent Step 3: Iterated commutators. \ 

We claim that for all $L\geq 1$, there are $C_L>0$, $N_L$ such that for all $\{b_i\}_{i=1}^L\subset S_{(1)}$, and $b\in S_{(1)}$, all $1\leq j\leq L$, and all $a\in S_{(0)}$ ($ A_t := 
e^{itP} a^{\rm{w}} e^{-i t P } $), 
\begin{equation}
\label{e:multipleAd}
\|\ad_{b_j^{\rm w}}\dots \ad_{b_1^{\rm w}} A_t\|_{L^2\to L^2}\leq C_Le^{C_L|t|}\sum_{0\leq |\alpha|\leq N_L}\|\partial^\alpha a\|_{\infty}\prod_{\ell=1}^j\sum_{1\leq |\beta|\leq N_L}\|\partial^\beta b_\ell\|_{\infty}
\end{equation}
and for $t\in\mathbb{R}$, and $0\leq j\leq L$,
\begin{equation}
\label{e:QCommutator}
\begin{split}
& \|\ad_{(b_j\circ\varphi_{-t})^{\rm w}}\dots \ad_{(b_1\circ \varphi_{-t})^{\rm w}} Q(b;-t)\|_{L^2\to L^2} \\
& \ \ \ \ \ \ \  \leq C_L e^{C_L|t|}\sum_{1\leq |\alpha|\leq N_L}\|\partial^\alpha b\|_{\infty}\prod_{\ell=1}^j\sum_{1\leq |\beta|\leq N_L}\|\partial^\beta b_\ell\|_{\infty}.
\end{split}
\end{equation}
Since~\eqref{e:multipleAd} implies~\eqref{e:beals} the proof of the lemma will be complete once that inequality is established. 

For $ L =1 $ we have proved~\eqref{e:multipleAd} and~\eqref{e:QCommutator} in steps 1 and 2 respectively. Therefore, we suppose by induction that~\eqref{e:multipleAd} and~\eqref{e:QCommutator} hold for some $L\geq 1$. Let 
$\{b_i\}_{i=1}^{L+1}\in S_{(1)}$,  set 
$$B_{i,1}:=(b_i\circ\varphi_{-t})^{\rm w},\qquad B_{i,0}:=-Q(b_i;-t),$$ 
and consider 
\begin{equation}
\label{e:expanded}
\begin{aligned}
e^{-itP}(\ad_{b_{L+1}^{\rm w}}\dots \ad_{b_1^{\rm w}} A_t)e^{itP}&=\ad_{B_{L+1,1}+B_{L+1,0}}\dots \ad_{B_{1,1}+B_{1,0}}a^{\rm w}\\
&=\sum_{\vec{\sigma}\in \{0,1\}^{L+1}}\ad_{B_{L+1,\vec{\sigma}_{L+1}}}\dots \ad_{B_{1,\vec{\sigma}_1}}a^{\rm w},
\end{aligned}
\end{equation}
where we used \eqref{e:singleCommutatorForm} to obtain the first equality. The second equality follows from expanding $ \ad_{ B_{j,1} + B_{j,0} } = \ad_{ B_{j,1} } + \ad_{ B_{j,0} }$. 

The Jacobi identity, $\ad_{A_1}\ad_{A_2}= \ad_{A_2}\ad_{A_1}+\ad_{[A_1,A_2]}$, 
allows us to rewrite this so that all the terms with $ \sigma_j = 1 $ are on the right:
\begin{equation}
\label{e:reordered}
\ad_{B_{L+1,\vec{\sigma}_{L+1}}}\dots \ad_{B_{1,\vec{\sigma}_1}}a^{\rm w}=\sum_{\gamma\in \Gamma_{\vec{\sigma}}} c_{\gamma}\ad_{\widetilde{Q}_{\gamma,1}}\dots \ad_{\widetilde{Q}_{\gamma,R_\gamma}}\ad_{B_{j_{\gamma,1},1}}\dots \ad_{B_{j_{\gamma,J_\gamma},1}}a^{\rm w},
\end{equation}
where for a fixed $ \vec \sigma \in \{ 0, 1 \}^{L+1}$, $ \Gamma_{\vec \sigma } $ is the set of
terms obtained after this reordering, with $ \gamma $ labelling its elements, 
$c_{\gamma}\in\{1,-1\}$, and 
$$
\widetilde{Q}_{\gamma,j}:=\ad_{B_{\ell_{\gamma,j,1},1}}\ad_{B_{\ell_{\gamma,j,2},1}}\dots\ad_{B_{\ell_{\gamma,j,M_{\gamma,j}},1}}Q(b_{i_{\gamma,j}};-t). 
$$
Here $0\leq M_{\gamma,j}$ (with the convention that when $ M_{\gamma, j } = 0$, 
$ \widetilde{Q}_{\gamma,j} = Q( b_{ i_{\gamma, j }} ; -t $) and 
$$R_{\gamma}+\sum_{j=1}^{R_{\gamma}}M_{\gamma,j}+J_{\gamma}=L+1. $$
Each $b_{i}$ appears exactly once in each summand in ~\eqref{e:reordered}.

To estimate the terms in \eqref{e:reordered}, we first see that~\cite[Lemma A.1]{gaz5}, \cite[Theorem 4.23]{z12}, and Lemma \ref{l:flow} imply there is $N_0$ such that 
\[ \begin{split}
& \|\ad_{B_{j_{\gamma,1},1}}\dots \ad_{B_{j_{\gamma,J_\gamma},1}}a^{\rm w}\|_{L^2\to L^2} \leq \sum_{|\beta|\leq J_{\gamma}+N_0}\|\partial^\beta a\|_{\infty}\prod_{\ell=1}^{J_{\gamma}}\sum_{1\leq |\beta|\leq N_0+J_{\gamma}}\|\partial^\beta b_{j_{\gamma,\ell}}\circ\varphi_{-t}\|_{\infty}\\
&\ \ \ \ \ \ \ \ \ \ \ \ \ \ \ \ \ \ \ \ \ \ \ \ \ \ \ \ \ \leq Ce^{J_\gamma(N_0+J_\gamma)\Lambda |t|}\sum_{|\beta|\leq J_{\gamma}+N_0}\|\partial^\beta a\|_{\infty}\prod_{\ell=1}^{J_{\gamma}}\sum_{1\leq |\beta|\leq N_0+J_\gamma}\|\partial^\beta b_{j_{\gamma,\ell}}\|_{\infty}.
\end{split} \]
Then, the inductive hypothesis~\eqref{e:QCommutator} implies that 
\[
\begin{split} 
\|\ad_{\widetilde{Q}_{\gamma,j}}A\|_{L^2\to L^2} & \leq 2 \| \widetilde{Q}_{\gamma,j} \|_{L^2 \to L^2}  \|A \|_{L^2 \to L^2} \\
& \leq 
2C_Le^{C_L|t|}\|A\|_{L^2\to L^2}\sum_{1\leq |\alpha|\leq N_L}\|\partial^\alpha b_{i_{\gamma,j}}\|_{\infty}\prod_{m=1}^{M_{\gamma,j}}\sum_{1\leq |\beta|\leq N_L}\|\partial^\beta b_{\ell_{\gamma,j,m}}\|_{\infty}.
\end{split} \]
Hence using the unitarity of $e^{itP}$, there are $\tilde{C}_{L+1}, \tilde{N}_{L+1}>0$ such that 
\begin{equation}
\label{e:weHaveAlmostInducted}
\begin{split} 
& \|\ad_{b_{L+1}^{\rm w}}\dots \ad_{b_1^{\rm w}} A_t\|_{L^2\to L^2}\\
& \ \ \ \ \ \ \leq \tilde{C}_{L+1}e^{\tilde{C}_{L+1}|t|}\sum_{|\beta|\leq \tilde{N}_{L+1}}\|\partial^\beta a\|_{\infty}\prod_{\ell=1}^{L+1}\sum_{1\leq |\beta|\leq \tilde{N}_{L+1}}\|\partial^\beta b_\ell\|_{\infty}.
\end{split} \end{equation}
In particular~\eqref{e:multipleAd} holds with $L$ replaced by $L+1$.

To prove~\eqref{e:QCommutator} with $L$ replaced by $L+1$, we write
\begin{align*}
&\|\ad_{(b_{L+1}\circ\varphi_{-t})^{\rm w}}\dots \ad_{(b_1\circ \varphi_{-t})^{\rm w}} Q(b;-t)\|_{L^2\to L^2}\\
&\leq \int_{\min(0,-t)}^{\max(-t,0)}\|\ad_{(b_{L+1}\circ\varphi_{-t})^{\rm w}}\dots \ad_{(b_1\circ \varphi_{-t})^{\rm w}} e^{i(-t-s)P}
q_s^{\rm w}e^{i(s+t)P}\|_{L^2\to L^2}ds
\end{align*}
where $q_s$ satisfies~\eqref{e:qEstimates}. 
Hence, applying~\eqref{e:weHaveAlmostInducted} with $a=q_s$, and $b_j$ replaced by $b_j\circ\varphi_{-t}$, we have
\begin{align*}
&\|\ad_{(b_{L+1}\circ\varphi_{-t})^{\rm w}}\dots \ad_{(b_1\circ \varphi_{-t})^{\rm w}} Q(b;-t)\|_{L^2\to L^2}\\
&\leq \int_{\min(0,-t)}^{\max(-t,0)}\tilde{C}_{L+1}e^{\tilde{C}_{L+1}|s+t|}\sum_{0\leq |\alpha|\leq \tilde{N}_{L+1}}\|\partial^\alpha q_s\|_{\infty}\prod_{\ell=1}^{L+1}\sum_{1\leq |\beta|\leq \tilde{N}_{L+1}}\|\partial^\beta b_\ell\circ\varphi_{-t}\|_{\infty}\\
&\leq C_{L+1}e^{C_{L+1}|t|}\sum_{1\leq |\alpha|\leq \tilde{N}_{L+1}+N}\|\partial^\alpha b\|_{\infty}\prod_{\ell=1}^{L+1}\sum_{1\leq |\beta|\leq \tilde{N}_{L+1}}\|\partial^\beta b_\ell\|_{\infty},
\end{align*}
which completes the proof of~\eqref{e:QCommutator} with $L$ replaced by $L+1$.

Taking $N_{L+1}\geq \max(\tilde{N}_{L+1}+N, N_L,N_0+L+1)$ and taking $C_{L+1}\geq \tilde{C}_{L+1}$,~\eqref{e:weHaveAlmostInducted} implies~\eqref{e:multipleAd} with $L$ replaced by $L+1$. Thus we have shown that~\eqref{e:multipleAd} and~\eqref{e:QCommutator} hold for all $L\geq 1$. 
\end{proof}

We now complete the proof of the Egorov theorem by reducing the proof for symbols in $S_{(1)}$ to that of $S_{(0)}$. 
\begin{proof}[Proof of Theorem~\ref{t:egorov}]
Observe that by Lemma~\ref{l:initialEgorov}
$$
e^{itP}b^{\rm w}e^{-itP}=(b\circ\varphi_t)^{\rm w}-i\int_0^t e^{i(t-s)P}q_s^{\rm w}e^{i(s-t)P}ds,
$$
where $q_s$ satisfies~\eqref{e:qEstimates}. By Lemma~\ref{l:zeroEgorov}, there is $q_{t,s}$ such that 
$$
e^{i(t-s)P}q_s^{\rm w}e^{i(s-t)P}=q_{t,s}^{\rm w},
$$
and 
\begin{align*}
\|q_{t,s}\|_{C^k}\leq C_ke^{C_k|t-s|}\|q_s\|_{C^{N_k}}&\leq C_ke^{C_k|t-s|+\Lambda (N_k+N)|s|}\sum_{1\leq |\alpha|\leq N_k+N}\|\partial^\alpha b\|_{\infty}\\
&\leq \tilde{C}_k e^{\tilde{C}_k(|t-s|+|s|)}.
\end{align*}
Hence we get 
\[ 
-i\int_0^t e^{i(t-s)P}q_s^{\rm w}e^{i(s-t)P}ds= r_t^{\rm w},
\ \ \ \ r_t:=-i\int_0^t q_{t,s}ds \]
where $ \|r_{t}\|_{C^k}\leq |t|\tilde{C}_k e^{\tilde{C}_k|t|} $.
This completes the proof.
\end{proof}

Proof of Theorem \ref{t:4} is now immediate: we apply Theorem \ref{t:egorov} and Lemma \ref{l:flow}. 

%


\appendix

\section{Some facts about pseudodifferential operators}
\label{a:a}

We first consider the general assumptions in \S \ref{s:ex1} and
use the results of \cite[\S 8.2]{z12} to see that there exists 
\begin{equation}
\label{eq:defG}  G = e^{ g^{\rm{w}}  ( x, D ) }  ,  \ \   G^s := e^{ s g^{\rm{w}} ( x, D ) } = e_s^{\rm{w}}  ( x, D ), \ \ 
e_s \in S ( m^{s } ), \ \ s \in \mathbb R . 
\end{equation}
From this we conclude that $ P G^{-1} = r ( x, D ) $ , $ r \in S( 1 ) $, and
that $ r(x, D)^{-1} = G P^{-1} : L^2 \to L^2  $ exists (since the domain of $ P $ was assumed to be 
$ H(m) $, $ P^{-1} : L^2 \to H(m) $ and $ G : H(m)  \to L^2 $). 
Beals's Lemma (see \cite[Theorem 8.3]{z12}) shows that
$ r(x,D)^{-1} = \tilde r( x, D) $, $ \tilde r \in S ( 1 ) $, and hence, 
\begin{equation}
\label{eq:Pinv}  P^{-1} = G^{-1} \tilde r ( x, D) = 
q ( x, D ),  \ \ \  q \in S ( m^{-1} ) . \end{equation}
  Then the assumptions in Theorem \ref{t:2} follow from 
\cite[Theorems 4.23 and 8.12]{z12}. 
In Example 3, we can compute $ c^{\rm{w}} ( x, D ) := [ P, A ] $, $ c_1^{\rm{w}} ( x, D ) :=  [ P , [P , A ] ]$ explicitly, and show that $ c \in S ( m^{\frac12} ) $ and $ c_1 \in S ( m^{ 1-\delta} ) $, 
$ m := ( 1 + |x| + |\xi| )^2 $. 

Finally we comment on the condition under which $ e^{-P} $ is of trace class. In the general case of
a self-adjoint 
$ P$ satisfying \eqref{eq:asP1} (with the domain $ H ( m ) $), $ P \geq 1 $, 
we  use \eqref{eq:Pinv} to conclude that 
$ P^{-\ell}  = q_\ell^{\rm{w}} ( x,D ) $, $ q_\ell \in  S( m^{-\ell } ) $. 
If $ m^{-\ell} \in L^2( \mathbb R^{2n} , dx d\xi ) $ then  $ P^{-\ell } $ is a Hilbert--Schmidt operator 
(see \eqref{eq:defHS} below) and consequently 
$ P^{-2\ell } $ is of trace class. Since for $ P \geq 0 $,   $\exp (-P ) \leq (2 \ell ) ! P^{-2\ell} $, it follows that
\(e^{-P}\) is of trace class.

\section{Localised generator for the harmonic oscillator}
\label{a:b}

The appendix shows how the localised generators (and the Davies generator) work for the
harmonic oscillator and $ \mathcal A $ consisting of linear operators. In this case 
general operator-algebraic and quantum-information frameworks collapse
to an explicit phase-space computation: exact Egorov and the exact Weyl
calculus for quadratic and linear symbols give a 
Fokker--Planck operator. For additional simplicity, we present this in the case one dimension (resulting in a two dimensional 
Fokker--Planck operator) but the same methods work in all dimensions. We also stress
that quantum
Ornstein--Uhlenbeck semigroups and harmonic-oscillator Gibbs samplers
have been studied in much greater generality in
\cite{cipr},\cite{beck} and \cite{smid}. Thus the appendix should only be read as a
comparison and normalisation check in an explicit setting.

\subsection{The localised construction}

We take
\begin{equation}
\label{eq:harm} 
 P=D_x^2+x^2 = p^{\rm{w}} ( x, D ) , \ \  p = \xi^2 + x^2, \ \ \
  \mathcal A=\{x,D_x\}.
\end{equation}
For the harmonic oscillator the Hamiltonian flow is given by 
$ x ( t ) = x (0 )\cos 2t + \xi ( 0) \sin 2t $, $ \xi ( t ) = \xi ( 0 ) \cos 2t - x ( 0) \sin 2t $, and 
Egorov's theorem is exact \cite[Lemma 11.8, Theorem 11.9]{z12}. Hence, 
\begin{equation*}
 e^{itP}xe^{-itP}=x\cos 2t+D_x\sin 2t, \ \ \  e^{itP}D_xe^{-itP}=D_x\cos 2t-x\sin 2t .
\end{equation*}
It is convenient to use the annihilation and creation operators
\begin{equation*}
 a=x+iD_x, \ \   a^*=x-iD_x , \ \ \  e^{itP}ae^{-itP}=e^{-2it}a, \ \ 
 e^{itP}a^*e^{-itP}=e^{2it}a^* .
\end{equation*}
Consequently the dissipative part of $\mathcal L_f$ becomes
\begin{equation}
 \mathcal D_f(T)=\kappa_-
 \left(aTa^* -\frac12\{a^*a,T\}\right)
 +\kappa_+
 \left(a^*Ta -\frac12\{aa^*,T\}\right),
\end{equation}
where
\begin{equation}
 \kappa_\pm :=\tfrac12\int_{\RR}\gamma(\omega)|\widehat f(\omega\mp2)|^2\,d\omega.
\end{equation}
The coherent correction
$B$ vanishes for \eqref{eq:harm}. The balance condition which makes $e^{-P}$ stationary gives
\begin{equation}
\label{eq:apba}
 {\kappa_-}/{\kappa_+}=e^2 .
\end{equation}
We should note that in the quadratic case \eqref{eq:harm} the Davies generator has the same structure:
in view of Proposition \ref{p:loc2Da} (or by a direct calculation) we obtain the same formula but now with
$  \kappa_\pm   = \frac 12 \gamma ( \pm 2 ) $. 

\subsection{Action on Weyl symbols}

Let $T=a^w(x,D_x)$ be the Weyl quantization of $a=a(x,\xi)$, with $h=1$.
Since $p$ is quadratic and the symbols of $a_\pm$ are linear, the
composition formula \cite[Theorem 4.18]{z12} has no remainders. In the terminology of 
\cite[\S 1.1]{gaz5} that means that the quantum (Lindblad) evolution agrees with the Fokker--Planck evolution:
\begin{equation}
\label{eq:L2Q} 
 \mathcal L_f(a^w)=(Q_fa)^w ,
\end{equation}
where 
\begin{equation}
\label{eq:Qf}
\begin{gathered}
 Q_f=-\beta^{-1}H_p+
 \delta(x\partial_x+\xi\partial_\xi+2)+D(\partial_x^2+\partial_\xi^2), \\
 \delta:=\kappa_- -\kappa_+, \qquad D:= \tfrac12 ( \kappa_-+\kappa_+) . 
 \end{gathered} 
\end{equation}
The Hamilton vector field, $ H_p = 2 \partial \xi \partial_x - 2 x \partial \xi $, generates
 a rotation, and the other term is an Ornstein--Uhlenbeck operator on $\RR^2_{x,\xi}$.
One checks that
\begin{equation}
 \rho(x,\xi)=C e^{-\alpha(x^2+\xi^2)}, \qquad
 \alpha=\frac{\delta}{2D}=\frac{\kappa_- -\kappa_+}{\kappa_-+\kappa_+} ,
\end{equation}
satisfies $ Q_f \rho = 0 $. 
In view of the balance condition \eqref{eq:apba},  $ \alpha=\tanh 1$, 
and $ \rho ( x, \xi ) $ is the Weyl symbol of the Gibbs state:
\begin{equation}
\label{eq:Meh}
 \rho(x,\xi)= 2 \tanh(1) e^{-\tanh(1)(x^2+\xi^2)} , \ \ \ 
 e^{-P}/ {\rm{tr}}\, e^{-P} = \rho^{\rm{w}} ( x, D )  . 
\end{equation}
(This follows from Mehler's formula and can be seen formally from \cite[Lemma 11.8, Theorem 11.9]{z12}
by taking imaginary time.)  

\subsection{Fokker--Planck operator on weighted spaces}
\label{s:weight}
The non-self-adjoint operator $ Q_f $ (defined in the maximal domain -- see \cite[Proposition A.2]{gaz5}) does not have a spectral gap on $ L^2 ( \mathbb R^2 ) $ -- its spectrum can be shown to be 
$ \{ z \in \mathbb C : \Re z \leq 0  \} $ (it is easiest to see on the Fourier transform side and then in polar coordinates). We note here $ L^2 $ on symbols corresponds the Hilbert--Schmidt norm on operators:
\begin{equation}
\label{eq:defHS}  \| a^{\rm{w}} ( x, D ) \|_{\rm{HS}}^2 = \frac{1}{ 2  \pi} \int_{ \mathbb R^2 } | a ( x, \xi ) |^2 dx d\xi =: \| a \|_{ \rm{HS}} ^2 .\end{equation}
This is the $L^2$-norm we will use in this section. 

To obtain an operator with discrete spectrum 
we conjugate $ Q_f $  by the square root of the equilibrium density:
\begin{equation}
 \widetilde Q_f:=\rho^{-\frac12}Q_f\rho^{\frac12},
\end{equation}
\begin{equation}
 \widetilde Q_f=-\beta^{-1}H_p+D(\partial_x^2+\partial_\xi^2)
 -\frac{\delta^2}{4D}(x^2+\xi^2)+\delta .
\end{equation}
The discreteness of the spectrum comes from the compactness of the resolvent of $ \widetilde Q_f $ --
see \cite[Theorem 4.28]{z12}. 

In polar coordinates, 
$r^2=x^2+\xi^2$, $ \tan \theta = \xi/x $, and with $\alpha=\delta/(2D)$,
\begin{equation}
\label{eq:Qpolar} 
 \widetilde Q_f= 2 \beta^{-1} \partial_\theta - D\left( - \partial_r^2 - \frac1 r  \partial_r - 
 \frac1 {r^2} \partial_\theta^2 +\alpha^2r^2 \right)+\delta .
\end{equation}
This representation immediately shows that $ \widetilde Q_f $ is normal and then that the eigenfunctions are 
\begin{equation}
 \varphi_{q,m}(r,\theta)=C_{q,m}r^{|m|}L_q^{|m|}(\alpha r^2)
 e^{-\alpha r^2/2}e^{im\theta}, \qquad q\in\NN_0,\quad m\in\ZZ,
\end{equation}
where $L_q^{|m|}$ denotes the associated Laguerre polynomial. The corresponding 
eigenvalues are 
\begin{equation}
 \lambda_{q,m}=-\delta(2q+|m|)+2i\beta^{-1}m,
 \qquad q\in\NN_0,\quad m\in\ZZ .
\end{equation}
For finite $\beta$, these eigenvalues are simple. If the rotation term is removed,
then the spectrum collapses to $\{-\delta N:N\in\NN_0\}$, with multiplicity $N+1$.
It follows that $0$ is a simple eigenvalue and the rest of the spectrum satisfies
\begin{equation}
 \operatorname{Re}\lambda\leq -\delta .
\end{equation}
 The $HS$-normalised (see \eqref{eq:defHS}) ground state is given by 
\begin{equation}
\label{eq:ground}  \varphi_{0} ( r ) = \varphi_{0,0} ( r ) = \sqrt{ 2\alpha }   e^{ - \alpha r^2/2 } = 
  \rho, \ \ \ \alpha = \tanh ( 1) . 
\end{equation}
(We will see an easier derivation of the spectral properties  
in \S \ref{s:FBI}.)

The spectral gap is
\begin{equation}
 \delta=\kappa_- -\kappa_+ .
\end{equation}
Since $\widetilde Q_f$ is normal, the spectral theorem gives
\begin{equation}
\label{eq:decaytQ}
 \|e^{t\widetilde Q_f} u - \langle u, \varphi_0 \rangle_{\rm{HS}}  \varphi_0 \|_{L^2(\RR^2)}\leq e^{-\delta t}\|u
 -  \langle u, \varphi_0 \rangle \varphi_0 \|_{L^2(\RR^2)} .
\end{equation}

\subsection{A Fokker--Planck operator as a quantum Gibbs sampler}
We now address \eqref{eq:Lbeta} for $ \mathcal L =  \mathcal L_f $ or equivalently, in view of 
\eqref{eq:L2Q}, for the Fokker--Planck evolution: if $ A = a^{\rm{w}} ( x, D ) $ and 
$ B = b^{\rm{w}} ( x, D ) $ then 
\[  {\rm{tr}} \,  A e^{ t \mathcal L_f } B = {\rm{tr}} \,  a^{\rm{w}} ( x, D ) ( e^{t Q_f} b) ^{\rm{w}} ( x, D ) =
\frac{1}{ 2 \pi} \int_{\mathbb R^2}  a ( x, \xi ) (e^{ t Q_f } b ) ( x, \xi ) d x d \xi ,\]
and, 
\begin{equation}
\label{eq:B2Gibbs}  \begin{split}  {\rm{tr}} \,  A e^{ t \mathcal L_f } B - \langle A \rangle & = 
 {\rm{tr}} \,  ( A ( e^{ t \mathcal L_f } B - e^{-P} /  {\rm{tr}}  e^{-P} ) \\ & = 
 \frac{1}{ 2 \pi} \int_{\mathbb R^2}  a ( x, \xi )(  (e^{ t Q_f } b ) ( x, \xi ) - c(b) \rho ( x, \xi ) ) d x d \xi ,
 \end{split} 
  \end{equation}
 where $ \rho $ was given in \eqref{eq:Meh}.
 
 To analyse convergence we translate \eqref{eq:decaytQ} to an estimate for $ e^{t Q_f } $ on the weighted space $L^2 ( \rho^{-1} dx d \xi ) $. For that we write:
 \[ e^{ t Q_f } b - c(b)  \rho = \rho^{\frac12} ( e^{ t \widetilde Q_f }( \rho^{-\frac12} b) -   c(b)  \varphi_0  ) , \]
 and hence, with 
 \[  c ( b ) =\langle \rho^{-\frac12 } b , \varphi_0 \rangle_{\rm{HS}} =
 \frac{1}{ 2 \pi}  \int_{\mathbb R^2 } b ( x, \xi ) d x d \xi = \tr B . \]
 (Once convergence is established it also follows from the fact $ \tr e^{ t Q_f } B = \tr B $.)
  
Now, from \eqref{eq:decaytQ}:
\begin{equation}
 \| e^{ t Q_f } b - \tr b^{\rm{w}} \,  \rho\|_{L^2(\rho^{-1}\,dx\,d\xi)}
 \leq e^{-\delta t}
 \|b - \tr b^{\rm{w}} \,  \rho\|_{L^2(\rho^{-1}\,dx\,d\xi)} , 
\end{equation}
and hence, returning to \eqref{eq:B2Gibbs},  if $ \tr B = 1 $ then
\begin{equation}
\label{eq:decayQ}  | {\rm{tr}} \,  A e^{ t \mathcal L_f } B - \langle A \rangle | \leq 
e^{- \delta t } \| a \|_{ L^2 (\rho dx d \xi ) } \| b - \rho \|_{ L^2( \rho^{-1} dx d\xi )} .
\end{equation}

We note that $ b ( x, \xi )  $ has to be very localised for it to be square integrable with respect to 
$ \rho ( x, \xi )^{-1} dx d \xi = (2 \tanh(1))^{-1} e^{ \tanh (1) ( x^2 + \xi^2 ) } dx d\xi $. On the other hand observables
$ A = a^{\rm{w}} ( x, \xi ) $ can belong to very general classes as square integrability is now required
only with respect to $  e^{ -\tanh (1) ( x^2 + \xi^2 ) } dx d\xi$. We also remark that we did not insist on $
b ^{\rm w}$ being of trace class: once $ b \in L^2 ( \rho^{-1} dx d\xi ) $ then $ b \in L^1 ( dx d\xi ) $ and the formal 
trace $ ( 2\pi )^{-1} \int b ( x, \xi ) dx d \xi $ is well defined. 

We conclude this section with two remarks:

\noindent 1. The resolvent of 
\[ ( Q_f - z )^{-1}  = \rho^{-\frac12} ( \widetilde Q_f - z )^{-1} 
\rho^{\frac12} : L^2 ( \rho^{-1} dx d\xi ) \to   L^2 ( \rho^{-1} dx d\xi ) , \]
is meromorphic, with poles at the eigenvalues of $ \widetilde Q_f $. The principal pole occurs at 
$ 0 $ and the singular part of the resolvent comes from the corresponding singular part of 
$ (  \widetilde Q_f - z )^{-1} $,  $ v \mapsto  - z^{-1} \langle v, \varphi_0 \rangle_{HS}  \varphi_0 $:
\[  ( Q_f - z )^{-1} = -   z^{-1} \rho \otimes 1   + R ( z ) , \ \ \ \  ( \rho \otimes 1 ) ( v ):= \rho ( x, \xi ) \frac{1}{ 2 \pi }
\int_{\mathbb R^2}  v ( y, \eta ) d y d \eta , \]
where $ R ( z ) $ is holomorphic in $ \Re z > -  \delta $. 

\noindent 2. The $L^2$--norm with respect to $ \rho^{-1} dx d \xi $ has similar features as 
the inverse Kubo--Martin--Schwinger (KMS) norm used in, say, \cite{chen}. For pseudodifferential operators 
it is defined as
\[ \| a^{\rm{w}} ( x, D ) \|_{\rm{KMS},* }^2 := \tr \left( (\rho^{\rm{w}} )^{-\frac12} a^{\rm{w}}  (\rho^{\rm{w}} )^{-\frac12}
( a^{\rm{w}}  )^* \right) = \|  (\rho^{\rm{w}} )^{-\frac14} a^{\rm{w}}  (\rho^{\rm{w}} )^{-\frac14} \|_{\rm{HS}}^2 .  \]
(For classes of $ a's $ with sufficient decay at infinity for this to make sense).

\subsection{The spectral calculation on the FBI side}
\label{s:FBI}

The calculation of the spectrum at the end of \S \ref{s:weight} becomes easier on the FBI
transform side -- see \cite[Chapter 13]{z12}.
To see that we put
\begin{equation}
 y=(x,\xi)\in\RR^2, \qquad z=(z_1,z_2)\in\CC^2 .
\end{equation}
We then choose the Bargmann/FBI transform so that 
the harmonic oscillator
\begin{equation}
 -\Delta_y+\alpha^2 y^2 = \alpha^2 ( - h^2 \Delta + y^2 ) , \ \ \ h = 1/\alpha, 
\end{equation}
takes a simple form  -- see \cite[Example (ii), \S  13.5.3]{z12}. We recall that the transform introduced there 
$ T_\alpha : L^2 ( \mathbb R^2, dy ) \to H_{\Phi} ( \mathbb C^2 )  $, 
\[ 
 H_{\Phi} ( \mathbb C^2 ) :=\left\{u\in \mathscr O (\CC^2):
 \|u\|_{\Phi}^2 :=  \int_{\CC^2}|u(z)|^2e^{-\alpha \Phi( z) } \,dm(z) <\infty\right\},  \
 \ \  \Phi (z):=\tfrac{1}{2} |z|^2 . 
\] 
is unitary. Since, 
\[ 
 T_\alpha(-\Delta_y+\alpha^2|y|^2)T_\alpha ^*
 =2\alpha\langle z,\partial_z\rangle+2\alpha,
\ \ \ \  \langle z,\partial_z\rangle=z_1\partial_{z_1}+z_2\partial_{z_2} .
 , \]
and 
for $H_p$,
\begin{equation*}
 T_\alpha H_p T_\alpha^*=2(z_2\partial_{z_1}-z_1\partial_{z_2}),
\end{equation*}
we have 
\begin{equation*}
 B_f:= T_\alpha \widetilde Q_f T_\alpha^* =  -\delta\langle z,\partial_z\rangle
 -2\beta^{-1}(z_2\partial_{z_1}-z_1\partial_{z_2}) . 
\end{equation*}
To compute the spectrum we write
\begin{equation*}
 z_+=z_1+iz_2, \qquad z_-=z_1-iz_2 , \ \ \ \ 
 \langle z,\partial_z\rangle z_+^jz_-^k=(j+k)z_+^jz_-^k,
\end{equation*}
and
\begin{equation*}
 (z_2\partial_{z_1}-z_1\partial_{z_2})z_+^jz_-^k=-i(j-k)z_+^jz_-^k .
\end{equation*}
Hence
\begin{equation*}
 B_f z_+^jz_-^k=\left(-\delta(j+k)+2i\beta^{-1}(j-k)\right)z_+^jz_-^k,
 \qquad j,k\in\NN_0 .
\end{equation*}
Thus the spectrum is
\begin{equation}
 \lambda_{j,k}=-\delta(j+k)+2i\beta^{-1}(j-k),
 \qquad j,k\in\NN_0 .
\end{equation}
This is the same list as in the polar-coordinate calculation. The translation between
the two parametrisations is
\begin{equation}
 N=j+k=2q+|m|, \qquad m=j-k .
\end{equation}
Thus the Laguerre modes in the real picture are replaced by monomials in the FBI
picture. In particular, $1$ is the unique eigenfunction for the eigenvalue $0$, and all
other eigenvalues satisfy $\operatorname{Re}\lambda\leq -\delta$.

\end{document}